\documentclass[12pt]{amsart}
\usepackage{amssymb,color}
\numberwithin{equation}{section} 

\usepackage{graphicx}

\newcommand{\bea}{\begin{eqnarray}}
\newcommand{\eea}{\end{eqnarray}}
\newcommand{\ba}{\begin{array}}
\newcommand{\ea}{\end{array}}
\newcommand{\edc}{\end{document}}
\newcommand{\bc}{\begin{center}}
\newcommand{\ec}{\end{center}}
\newcommand{\be}{\begin{equation}}
\newcommand{\ee}{\end{equation}}

\def\cb{{\mathcal B}}

\def\cg{{\mathcal G}}

\def\bc{{\mathbb C}}

\def\bn{{\mathbb N}}

\def\bq{{\mathbb Q}}

\def\bz{{\mathbb Z}}

\def\g{\gamma}  \def\G{\Gamma}
\def\d{\delta}  

\def\e{\epsilon}
\def\l{\lambda} 

\def\m{\mu}

\def\n{\nu}
\def\r{\rho}
\def\s{\sigma} 
\def\t{\theta}

\def\w{\omega} \def\Om{\Omega}

\def\h{{\mathbf{h}}}

\def\xb{{\mathbf{x}}}
\def\sb{{\mathbf{s}}}

\newtheorem{thm}{Theorem}[section]
\newtheorem{lem}[thm]{Lemma}

\newtheorem{prop}[thm]{Proposition}

\theoremstyle{remark}
\newtheorem{rem}{Remark}[section]

\begin{document}

\title[On $p$-adic Gibbs measures]
{On dynamical systems and phase transitions for $Q+1$-state $P$-adic
Potts model on the Cayley tree}


\author{Farrukh Mukhamedov}
\address{Farrukh Mukhamedov\\
 Department of Computational \& Theoretical Sciences\\
Faculty of Science, International Islamic University Malaysia\\
P.O. Box, 141, 25710, Kuantan\\
Pahang, Malaysia} \email{{\tt far75m@yandex.ru} {\tt
farrukh\_m@iiu.edu.my}}

\begin{abstract}
In the present paper, we introduce a new kind of $p$-adic measures
for $q+1$-state Potts model, called {\it $p$-adic quasi Gibbs
measure}. For such a model, we derive a recursive relations with
respect to boundary conditions. Note that we consider two mode of
interactions: ferromagnetic and antiferromagnetic. In both cases, we
investigate a phase transition phenomena from the associated
dynamical system point of view. Namely, using the derived recursive
relations we define one dimensional fractional $p$-adic dynamical
system. In ferromagnetic case, we establish that if $q$ is divisible
by $p$, then such a dynamical system has two repelling and one
attractive fixed points. We find basin of attraction of the fixed
point. This allows us to describe all solutions of the nonlinear
recursive equations. Moreover, in that case there exists the strong
phase transition. If $q$ is not divisible by $p$, then the fixed
points are neutral, and this yields that the existence of the quasi
phase transition. In antiferromagnetic case, there are two
attractive fixed points, and we find basins of attraction of both
fixed points, and describe solutions of the nonlinear recursive
equation. In this case, we prove the existence of a quasi phase
transition.

\vskip 0.3cm \noindent {\it
Mathematics Subject Classification}: 46S10, 82B26, 12J12, 39A70, 47H10, 60K35.\\
{\it Key words}: $p$-adic numbers, Potts model; $p$-adic quasi
Gibbs measure, phase transition.
\end{abstract}

\maketitle

\section{introduction}

Due to the assumption that $p$-adic numbers provide a more exact and
more adequate description of microworld phenomena, starting the
1980s, various models described in the language of $p$-adic analysis
have been actively studied
\cite{ADFV},\cite{FO},\cite{MP},\cite{V1}. The well-known studies in
this area are primarily devoted to investigating quantum mechanics
models using equations of mathematical physics \cite{ADV,V2,VVZ}.
Furthermore, numerous applications of the $p$-adic analysis to
mathematical physics have been proposed in
\cite{ABK},\cite{Kh1},\cite{Kh2}. One of the first applications of
$p$-adic numbers in quantum physics appeared in the framework of
quantum logic in \cite{BC}. This model is especially interesting for
us because it could not be described by using conventional real
valued probability. Besides, it is also known
\cite{Kh2,Ko,MP,Ro,Vi,VVZ} that a number of $p$-adic models in
physics cannot be described using ordinary Kolmogorov's probability
theory. New probability models, namely $p$-adic ones were
investigated in \cite{BD},\cite{K3},\cite{KhN}. After that in
\cite{KYR} an abstract $p$-adic probability theory was developed by
means of the theory of non-Archimedean measures \cite{Ro}. Using
that measure theory in \cite{KL},\cite{Lu} the theory of stochastic
processes with values in $p$-adic and more general non-Archimedean
fields having probability distributions with non-Archimedean values
has been developed. In particular, a non-Archimedean analog of the
Kolmogorov theorem was proven (see also \cite{GMR}). Such a result
allows us to construct wide classes of stochastic processes using
finite dimensional probability distributions\footnote{We point out
that stochastic processes on the field $\bq_p$ of $p$-adic numbers
with values of real numbers have been studied by many authors, for
example, \cite{AK,AZ1,AZ2,DF,Koc,Y}. In those investigations wide
classes of Markov processes on $\bq_p$ were constructed and studied.
In our case the situation is different, since probability measures
take their values in $\bq_p$. This leads our investigation to some
difficulties. For example, there is no information about the
compactness of $p$-adic values probability measures. }. Therefore,
this result give us a possibility to develop the theory of
statistical mechanics in the context of the $p$-adic theory, since
it lies on the basis of the theory of probability and stochastic
processes. Note that one of the central problems of such a theory is
the study of infinite-volume Gibbs measures corresponding to a given
Hamiltonian, and a description of the set of such measures. In most
cases such an analysis depend on a specific properties of
Hamiltonian, and complete description is often a difficult problem.
This problem, in particular, relates to a phase transition of the
model (see \cite{G}).

In \cite{KK1,KK2} a notion of ultrametric Markovianity, which
describes independence of contributions to random field from
different ultrametric balls, has been introduced, and shows that
Gaussian random fields on general ultrametric spaces (which were
related with hierarchical trees), which were defined as a solution
of pseudodifferential stochastic equation (see also \cite{KaKo}),
satisfies the Markovianity. In addition,  covariation of the defined
random field was computed with the help of wavelet analysis on
ultrametric spaces (see also \cite{Koz}). Some applications of the
results to replica matrices, related to general ultrametric spaces
have been investigated in \cite{KK3}.

The aim of this paper is devoted to the development of  $p$-adic
probability theory approaches to study $q+1$-state nearest-neighbor $p$-adic Potts model
on Cayley tree (see \cite{W}). We are especially interested in the construction of $p$-adic quasi
Gibbs measures for the mentioned model, since such measures present
more natural concrete examples of $p$-adic Markov processes (see
\cite{KL}, for definitions). In \cite{MR1,MR2} we have studied  $p$-adic Gibbs
measures and existence of phase transitions for the $q$-state Potts models on the Cayley tree\footnote{The classical (real value)
counterparts of such models were considered in \cite{W}}. It was
established that a phase transition occurs \footnote{Here the phase
transition means the existence of two distinct $p$-adic Gibbs
measures for the given model.} if $q$ is divisible by $p$. This
shows that the transition depends on the number of spins $q$.

To investigate phase transitions, a dynamical system approach, in
real case, has greatly enhanced our understanding of complex
properties of models. The interplay of statistical mechanics with
chaos theory has even led to novel conceptual frameworks in
different physical settings \cite{E}. On the other hand, the theory
$p$-adic dynamical systems is a rapidly growing topic, there are
many papers devoted to this subject (see for example,
\cite{KhN},\cite{Sil1}).  We remark that first investigations of
non-Archimedean dynamical systems have appeared in \cite{HY}. We
also point out that intensive development of $p$-adic (and more
general algebraic) dynamical systems has happened few years, (for
example, see \cite{AKh,AV1,AV2,B1,FL1,FL2,KM1,RL,TVW,Wo}). More
extensive lists may be found in the $p$-adic dynamics bibliography
maintained by Silverman \cite{Sil2} and the algebraic dynamics
bibliography of Vivaldi \cite{Vi}.

In the present paper, we are going to investigate a phase transition
phenomena from the such a dynamical system point of view. In the
paper we introduce a new class of $p$-adic measures, associated with
$q+1$-state Potts model, called {\it $p$-adic quasi Gibbs measure}.
Note such a class is totaly different from the $p$-adic Gibbs
measures considered in \cite{MR1,MR2}. For the model under
consideration, we derive a recursive relations with respect to
boundary conditions. Note that we shall consider two mode of
interactions: ferromagnetic and antiferromagnetic.  Namely, using
the derived recursive relations we define one dimensional fractional
$p$-adic dynamical system. In both cases, we are going to
investigate a phase transition phenomena from the associated
dynamical system point of view. In ferromagnetic case, we establish
that if $q$ is divisible by $p$, then such a dynamical system has
two repelling and one attractive fixed points. We find basin of
attraction of the fixed point. This allows us to describe all
solutions of the nonlinear recursive equations. Moreover, in that
case there exists the strong phase transition. If $q$ is not
divisible by $p$, then the fixed points are neutral, and this yields
that the existence of the quasi phase transition. In
antiferromagnetic case, there are two attractive fixed points, and
we find basins of attraction of both fixed points, and describe
solutions of the nonlinear recursive equation. In this case, we
prove the existence of a quasi phase transition. Note that the
obtained results are totaly different from the results of
\cite{MR1,MR2}, since when $q$ is divisible by $p$ means that $q+1$
is not divided by $p$, which according to \cite{MR1} means that
uniqueness and boundedness of $p$-adic Gibbs measure.

\section{Preliminaries}

\subsection{$p$-adic numbers}

In what follows $p$ will be a fixed prime number, and $\bq_p$
denotes the field of $p$-adic filed, formed by completing $\bq$ with
respect to the unique absolute value satisfying $|p|_p = 1/p$. The
absolute value $|\cdot|_p$, is non- Archimedean, meaning that it
satisfies the ultrametric triangle inequality $|x + y|_p \leq
\max\{|x|_p, |y|_p\}$.

Any $p$-adic number $x\in\bq_p$, $x\neq 0$ can be uniquely represented in the form
\begin{equation}\label{canonic}
x=p^{\g(x)}(x_0+x_1p+x_2p^2+...),
\end{equation}
where $\g=\g(x)\in\bz$ and $x_j$ are integers, $0\leq x_j\leq p-1$,
$x_0>0$, $j=0,1,2,\dots$ In this case $|x|_p=p^{-\g(x)}$.

We recall that an integer $a\in \bz$ is called {\it a quadratic
residue modulo $p$} if the equation $x^2\equiv a(\textrm{mod $p$})$
has a solution $x\in \bz$.

\begin{lem}\label{quadrat} \cite{Ko} In order that the equation
$$
x^2=a, \ \ 0\neq a=p^{\g(a)}(a_0+a_1p+...), \ \ 0\leq a_j\leq p-1, \
a_0>0
$$
has a solution $x\in \bq_p$, it is necessary and sufficient that the
following conditions are fulfilled:
\begin{enumerate}
\item[(i)] $\g(a)$ is even;\\

\item[(ii)] $a_0$ is a quadratic residue modulo $p$ if $p\neq 2$, and moreover
$a_1=a_2=0$ if $p=2$.
\end{enumerate}
\end{lem}

Note the basics of $p$-adic analysis, $p$-adic mathematical physics
are explained in \cite{Ko,Ma,S,Ro,VVZ}.

\subsection{Dynamical systems in $\bq_p$}

In this subsection we recall some standard terminology of the theory
of dynamical systems (see for example \cite{PJS},\cite{KhN}).

Given $r,s>0$ ($r<s$) and $a\in\bq_p$ denote
\begin{eqnarray}\label{B}
&&B_r(a)=\{x\in\bq_p:\ |x-a|_p<r\}, \ \ \bar B_r(a)=\{x\in\bq_p:\ |x-a|_p\leq r\} \\
\label{S} && B_{r,s}(a)=\{x\in\bq_p:\ r<|x-a|_p<s\}, \ \
S_r(a)=\{x\in\bq_p:\ |x-a|_p=r\}.
\end{eqnarray}
It is clear that $\bar B_r(a)=B_r(a)\cup S_r(a)$.

A function $f:B_r(a)\to\bq_p$ is said to be {\it analytic} if it can
be represented by
$$
f(x)=\sum_{n=0}^{\infty}f_n(x-a)^n, \ \ \ f_n\in \bq_p,
$$ which converges uniformly on the ball $B_r(a)$.

Consider a dynamical system $(f,B)$ in $\bq_p$, where $f: x\in B\to
f(x)\in B$ is an analytic function and $B=B_r(a)$ or $\bq_p$. Denote
$x^{(n)}=f^n(x^{(0)})$, where $x^0\in B$ and
$f^n(x)=\underbrace{f\circ\dots\circ f(x)}_n$.
 If $f(x^{(0)})=x^{(0)}$ then $x^{(0)}$
is called a {\it fixed point}. A fixed point $x^{(0)}$ is called an
{\it attractor} if there exists a neighborhood $U(x^{(0)})(\subset
B)$ of $x^{(0)}$ such that for all points $y\in U(x^{(0)})$ it holds
$\lim\limits_{n\to\infty}y^{(n)}=x^{(0)}$, where $y^{(n)}=f^n(y)$.
If $x^{(0)}$ is an attractor then its {\it basin of attraction} is
$$
A(x^{(0)})=\{y\in \bq_p :\ y^{(n)}\to x^{(0)}, \ n\to\infty\}.
$$
A fixed point $x^{(0)}$ is called {\it repeller} if there  exists a
neighborhood $U(x^{(0)})$ of $x^{(0)}$ such that
$|f(x)-x^{(0)}|_p>|x-x^{(0)}|_p$ for $x\in U(x^{(0)})$, $x\neq
x^{(0)}$. For a fixed point  $x^{(0)}$ of a function $f(x)$ a ball
$B_r(x^{(0)})$ (contained in $B$) is said to be a {\it Siegel disc}
if each sphere $S_{\r}(x^{(0)})$, $\r<r$ is an invariant sphere of
$f(x)$, i.e. if $x\in S_{\r}(x^{(0)})$ then all iterated points
$x^{(n)}\in S_{\r}(x^{(0)})$ for all $n=1,2\dots$. The union of all
Siegel discs with the center at $x^{(0)}$ is said to {\it a maximum
Siegel disc} and is denoted by $SI(x^{(0)})$.

\begin{rem} In non-Archimedean geometry, a center of a
disc is nothing but a point which belongs to the disc, therefore, in
principle, different fixed points may have the same Siegel disc (see
\cite{AV2}).
\end{rem}

Let $x^{(0)}$ be a fixed point of an analytic function $f(x)$. Set
$$
\l=\frac{d}{dx}f(x^{(0)}).
$$

The point $x^{(0)}$ is called {\it attractive} if $0\leq |\l|_p<1$,
{\it indifferent} if $|\l|_p=1$, and {\it repelling} if $|\l|_p>1$.

\subsection{$p$-adic measure}

Let $(X,\cb)$ be a measurable space, where $\cb$ is an algebra of
subsets $X$. A function $\m:\cb\to \bq_p$ is said to be a {\it
$p$-adic measure} if for any $A_1,\dots,A_n\subset\cb$ such that
$A_i\cap A_j=\emptyset$ ($i\neq j$) the equality holds
$$
\mu\bigg(\bigcup_{j=1}^{n} A_j\bigg)=\sum_{j=1}^{n}\mu(A_j).
$$

A $p$-adic measure is called a {\it probability measure} if
$\mu(X)=1$.  A $p$-adic probability measure $\m$ is called {\it
bounded} if $\sup\{|\m(A)|_p : A\in \cb\}<\infty $. Note that in
general, a $p$-adic probability measure need not be bounded
\cite{K3,KL,Ko}. For more detail information about $p$-adic measures
we refer to \cite{K3},\cite{KhN},\cite{Ro}.

\subsection{Cayley tree}

Let $\Gamma^k_+ = (L,E)$ be a semi-infinite Cayley tree of order
$k\geq 1$ with the root $x^0$ (whose each vertex has exactly $k+1$
edges, except for the root $x^0$, which has $k$ edges). Here $L$ is
the set of vertices and $E$ is the set of edges. The vertices $x$
and $y$ are called {\it nearest neighbors} and they are denoted by
$l=<x,y>$ if there exists an edge connecting them. A collection of
the pairs $<x,x_1>,\dots,<x_{d-1},y>$ is called a {\it path} from
the point $x$ to the point $y$. The distance $d(x,y), x,y\in V$, on
the Cayley tree, is the length of the shortest path from $x$ to $y$.

Recall a coordinate structure in $\G^k_+$:  every vertex $x$ (except
for $x^0$) of $\G^k_+$ has coordinates $(i_1,\dots,i_n)$, here
$i_m\in\{1,\dots,k\}$, $1\leq m\leq n$ and for the vertex $x^0$ we
put $(0)$.  Namely, the symbol $(0)$ constitutes level 0, and the
sites $(i_1,\dots,i_n)$ form level $n$ ( i.e. $d(x^0,x)=n$) of the
lattice.

\begin{figure}
\begin{center}
\includegraphics[width=10.07cm]{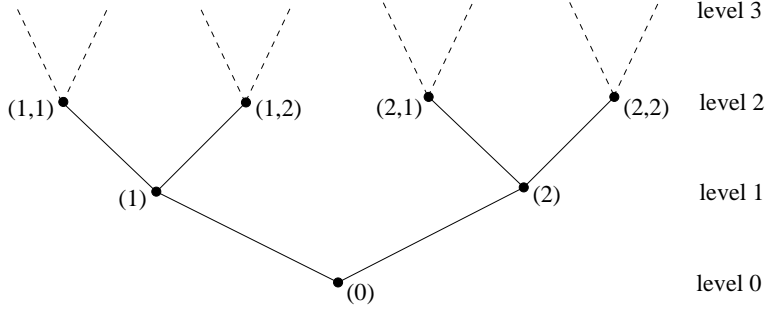}
\end{center}
\caption{The first levels of $\G_+^2$} \label{fig1}
\end{figure}

Let us set
$$ W_n=\{x\in V| d(x,x^0)=n\}, \ \ \
V_n=\bigcup_{m=1}^n W_m, \ \ L_n=\{l=<x,y>\in L | x,y\in V_n\}.
$$
For $x\in \G^k_+$, $x=(i_1,\dots,i_n)$ denote
\begin{equation}\label{S(x)}
 S(x)=\{(x,i):\ 1\leq
i\leq k\},
\end{equation}
here $(x,i)$ means that $(i_1,\dots,i_n,i)$. This set is called a
set of {\it direct successors} of $x$.

\section{$p$-adic Potts model and its $p$-adic quasi Gibbs measures}

In this section we consider the $p$-adic Potts model where spin takes values in the
set $\Phi=\{0,1,2,\cdots,q\}$, here $q\geq 1$, ($\Phi$ is called a
{\it state space}) and is assigned to the vertices of the tree
$\G^k=(V,\Lambda)$. A configuration $\s$ on $V$ is then defined as a
function $x\in V\to\s(x)\in\Phi$; in a similar manner one defines
configurations $\s_n$ and $\w$ on $V_n$ and $W_n$, respectively. The
set of all configurations on $V$ (resp. $V_n$, $W_n$) coincides with
$\Omega=\Phi^{V}$ (resp. $\Omega_{V_n}=\Phi^{V_n},\ \
\Omega_{W_n}=\Phi^{W_n}$). One can see that
$\Om_{V_n}=\Om_{V_{n-1}}\times\Om_{W_n}$. Using this, for given
configurations $\s_{n-1}\in\Om_{V_{n-1}}$ and $\w\in\Om_{W_{n}}$ we
define their concatenations  by
$$
(\s_{n-1}\vee\w)(x)= \left\{
\begin{array}{ll}
\s_{n-1}(x), \ \ \textrm{if} \ \  x\in V_{n-1},\\
\w(x), \ \ \ \ \ \ \textrm{if} \ \ x\in W_n.\\
\end{array}
\right.
$$
It is clear that $\s_{n-1}\vee\w\in \Om_{V_n}$.

The Hamiltonian $H_n:\Om_{V_n}\to\bq_p$ of the $p$-adic {\it
$q+1$-state Potts} model has a form
\begin{equation}\label{Potts}
H_n(\s)=N\sum_{<x,y>\in L_n}\delta_{\s(x),\s(y)},  \ \
\s\in\Om_{V_n}, \  n\in\mathbb{N},
\end{equation}
where $\delta$ is the Kronecker symbol and the coupling constant $N$ ($N\neq 0$),
belongs to $\bz$. We call the model {\it ferromagnetic} if $N>0$, and {\it antiferromagnetic} if $N<0$.

Note that when $q=1$, then the corresponding model reduces to the
$p$-adic Ising model. Such a model was investigated in \cite{GMR,KM}.

Now let us construct $p$-adic quasi Gibbs measures corresponding to the
model.

Assume that  $\h: V\setminus\{x^{(0)}\}\to\bq_p^{\Phi}$ is a function,
i.e. $\h_x=(h_{0,x},h_{1,x},\dots,h_{q,x})$,  where $h_{i,x}\in\bq_p$ ($i\in\Phi$) and $x\in V\setminus\{x^{(0)}\}$. Given $n\in\bn$, let us consider a $p$-adic probability measure $\m^{(n)}_\h$ on $\Om_{V_n}$ defined by
\begin{equation}\label{mu}
\mu^{(n)}_{\h}(\s)=\frac{1}{Z_n^{(\h)}}p^{H_n(\s)}\prod_{x\in
W_n}h_{\s(x),x}
\end{equation}
Here, $\s\in\Om_{V_n}$, and $Z_n^{(\h)}$ is the corresponding
normalizing factor called a {\it partition function} given by
\begin{equation}\label{ZN1}
Z_n^{(\h)}=\sum_{\s\in\Omega_{V_n}}p^{H_n(\s)}\prod_{x\in
W_n}h_{\s(x),x},
\end{equation}
here subscript $n$ and superscript $(\h)$ are accorded to the $Z$,
since it depends on $n$ and a function $\h$.

One of the central results of the theory of probability concerns a
construction of an infinite volume distribution with given
finite-dimensional distributions, which is called {\it Kolmogorov's
Theorem} \cite{Sh}. Therefore, in this paper we are interested in
the same question but in a $p$-adic context. More exactly, we want
to define a $p$-adic probability measure $\m$ on $\Om$ which is
compatible with defined ones $\m_\h^{(n)}$, i.e.
\begin{equation}\label{CM}
\m(\s\in\Om: \s|_{V_n}=\s_n)=\m^{(n)}_\h(\s_n), \ \ \ \textrm{for
all} \ \ \s_n\in\Om_{V_n}, \ n\in\bn.
\end{equation}

In general, \`{a} priori the existence such a kind of measure
$\m$ is not known, since there is not much information on
topological properties, such as compactness, of the set of all
$p$-adic measures defined even on compact spaces\footnote{In the
real case, when the state space is compact, then the existence
follows from the compactness of the set of all probability measures
(i.e. Prohorov's Theorem). When the state space is non-compact, then
there is a Dobrushin's Theorem \cite{Dob1,Dob2} which gives a
sufficient condition for the existence of the Gibbs measure for a
large class of Hamiltonians. }. Note that certain properties of the set of $p$-adic measures has been
studied in \cite{kas2}, but those properties are not enough to prove the existence of the limiting
measure. Therefore,
at a moment, we can only use the $p$-adic Kolmogorov extension
Theorem (see \cite{GMR},\cite{KL}) which based on so called {\it
compatibility condition} for the measures $\m_\h^{(n)}$, $n\geq 1$,
i.e.
\begin{equation}\label{comp}
\sum_{\w\in\Om_{W_n}}\m^{(n)}_\h(\s_{n-1}\vee\w)=\m^{(n-1)}_\h(\s_{n-1}),
\end{equation}
for any $\s_{n-1}\in\Om_{V_{n-1}}$. This condition according to the
theorem implies the existence of a unique $p$-adic measure $\m$
defined on $\Om$ with a required condition \eqref{CM}. Note that
more general theory of $p$-adic measures has been developed in \cite{kas1}.

So, if for some function $\h$ the measures  $\m_\h^{(n)}$ satisfy
the compatibility condition, then there is a unique $p$-adic
probability measure, which we denote by $\m_\h$, since it depends on
$\h$. Such a measure $\m_\h$ is said to be {\it a $p$-adic quasi
Gibbs measure} corresponding to the $p$-adic Potts model. By
$Q\cg(H)$ we denote the set of all $p$-adic quasi Gibbs measures
associated with functions $\h=\{\h_x,\ x\in V\}$. If there are at
least two distinct $p$-adic quasi Gibbs measures $\m,\n\in Q\cg(H)$
such that $\m$ is bounded and $\n$ is unbounded,
then we say that {\it a phase transition}
occurs. By another words, one can find two different functions $\sb$ and $\h$
defined on $\bn$ such that there exist the corresponding measures
$\m_\sb$ and $\m_\h$, for which one is bounded, another one is unbounded. Moreover, if
there is a sequence of sets $\{A_n\}$ such that $A_n\in\Om_{V_n}$  with
$|\m(A_n)|_p\to 0$ and $|\n(A_n)|_p\to\infty$ as $n\to\infty$, then we say that there occurs
a {\it strong phase transition}.
If there are two different functions $\sb$ and $\h$ defined on $\bn$ such that
there exist the corresponding measures $\m_\sb$, $\m_\h$, and they are bounded, then we
say there is a {\it quasi phase transition}.

\begin{rem} Note that in \cite{MR1} we considered the following
sequence of $p$-adic measures defined by
\begin{equation}\label{mr-mu}
\mu^{(n)}_{\h}(\s)=\frac{1}{\tilde Z_n^{(\h)}}\exp_p\{H_n(\s)\}\prod_{x\in
W_n}h_{\s(x),x},
\end{equation}
here as usual $\tilde Z_n^{(\h)}$ is the corresponding
normalizing factor. A limiting $p$-adic measures generated by \eqref{mr-mu} was called
{\it $p$-adic Gibbs measure}. Such kind of measures and phase transitions, for Ising and Potts
models on Cayley tree, have been studied in \cite{GMR,KM,MR1,MR2}. When a state space $\Phi$ is countable, the
corresponding $p$-adic Gibbs measures have been investigated in \cite{KMM,M}.
\end{rem}

Now one can ask for what kind of functions $\h$ the measures
$\m_\h^{(n)}$ defined by \eqref{mu} would satisfy the
compatibility condition \eqref{comp}. The following theorem gives
an answer to this question.

\begin{thm}\label{comp1} The measures $\m^{(n)}_\h$, $
n=1,2,\dots$ (see \eqref{mu}) satisfy the compatibility condition
\eqref{comp} if and only if for any $n\in \bn$ the following
equation holds:
\begin{equation}\label{eq1}
\hat h_{x}=\prod_{y\in S(x)}{\mathbf{F}}(\hat \h_{y};\theta),
\end{equation}
here and below $\theta=p^N$, a vector $\hat \h=(\hat h_1,\dots,\hat
h_q)\in\bq_p^q$ is defined by a vector
$\h=(h_0,h_1,\dots,h_q)\in\bq_p^{q+1}$ as follows
\begin{equation}\label{H}
\hat h_i=\frac{h_i}{h_0}, \ \ \ i=1,2,\dots,q
\end{equation}
and mapping ${\mathbf{F}}:\bq_p^{q}\times\bq_p\to\bq_p^q$ is defined
by ${\mathbf{F}}(\xb;\t)=(F_1(\xb;\t),\dots,F_q(\xb;\t))$ with
\begin{equation}\label{eq2}
F_i(\xb;\theta)=\frac{(\theta-1)x_i+\sum\limits_{j=1}^{q}x_j+1}
{\sum\limits_{j=1}^{q}x_j+\theta}, \ \ \xb=\{x_i\}\in\bq_p^q, \ \
i=1,2,\dots,q.
\end{equation}
\end{thm}

The proof consists of checking condition \eqref{comp} for the
measures \eqref{mu} (cp. \cite{MR1,KMM}).

\begin{lem}\label{parti} Let $\h$ be a solution of \eqref{eq1}, and
$\m_\h$ be an associated $p$-adic quasi Gibbs measure. Then for the
corresponding partition function $Z^{(\h)}_n$ (see \eqref{ZN1})
the following equality holds
\begin{equation}\label{ZN2}
Z^{(\h)}_{n+1}=A_{\h,n}Z^{(\h)}_n,
\end{equation}
where $A_{\h,n}$ will be defined below (see \eqref{aN3}).
\end{lem}

\begin{proof} Since $\h$ is a solution of \eqref{eq1}, then we conclude that there is a constant
$a_\h(x)\in\bq_p$ such that
\begin{equation}\label{aN1}
\prod_{y\in
S(x)}\sum_{j=0}^qp^{N\d_{ij}}h_{j,y}=a_{\h}(x)h_{i,x}
\end{equation}
for any $i\in\{0,\dots,q\}$. From this one gets
\begin{eqnarray}\label{aN2}
\prod_{x\in W_{n}}\prod_{y\in
S(x)}\sum_{j=0}^q p^{N\d_{ij}}h_{j,y}=\prod_{x\in
W_n}a_{\h}(x)h_{i,x}=A_{\h,n}\prod_{x\in W_n}h_{i,x},
\end{eqnarray}
where
\begin{equation}\label{aN3}
A_{\h,n}=\prod_{x\in W_n}a_{\h}(x).
\end{equation}
Given $j\in\Phi$, by $\eta^{(j)}\in\Om_{W_n}$ we denote a configuration on $W_n$ defined as follows: $\eta^{(j)}(x)=j$ for all $x\in W_n$.

Hence, by  \eqref{mu},\eqref{aN2} we have
\begin{eqnarray*}
1&=&\sum_{\s\in\Om_n}\sum_{\w\in\Om_{W_n}}\m^{(n+1)}_\h(\s\vee \w)\\
&=&\sum_{\s\in\Om_n}\sum_{\w\in\Om_{W_n}}\frac{1}{Z^{(\h)}_{n+1}}p^{H(\s\vee \w)}\prod_{x\in
W_{n+1}}h_{\w(x),x}\\
&=&\frac{1}{Z^{(\h)}_{n+1}}\sum_{\s\in\Om_n}p^{H(\s)}\prod_{x\in W_n}\prod_{y\in S(x)}\sum_{j=0}^q p^{N\d_{\s(x),j}}h_{j,y}\\
&=&\frac{A_{\h,n}}{Z^{(\h)}_{n+1}}
\sum_{\s\in\Om_n}p^{H(\s)}\prod_{x\in W_n}h_{\s(x),x}\\
&=&\frac{A_{\h,n}}{Z^{(\h)}_{n+1}}Z_n^{(\h)}
\end{eqnarray*}
which implies the required relation.
\end{proof}

\section{Dynamical systems and existence of $p$-adic quasi Gibbs measures}

In this section we will establish existence of $p$-adic quasi Gibbs
measures on a Cayley tree of order 2, i.e. $k=2$. To do it, we
reduce the equation \eqref{eq1} to the fixed point problem for
certain dynamical system. This allows us to investigate the
existence of $p$-adic quasi Gibbs measure.

We say that a function $\h=\{\h_x\}_{x\in V\setminus\{x^0\}}$ is
called {\it translation-invariant} if $\h_x=\h_y$ for all $x,y\in
V\setminus\{x^0\}$. A $p$-adic measure $\m_\h$, corresponding to a
translation-invariant function $\h$, is called {\it
translation-invariant $p$-adic quasi Gibbs measure}.

Let us first restrict ourselves to the description of
translation-invariant solutions of \eqref{eq1}, namely
$\h_x=\h(=(h_0,h_1,\dots,h_q))$ for all $x\in V$. Then \eqref{eq1}
can be rewritten as follows
\begin{equation}\label{eq11}
\hat h_{i}=\bigg(\frac{(\theta-1)\hat h_{i}+\sum_{j=1}^{q}\hat
h_j+1} {\sum_{j=1}^{q}\hat h_j+\theta}\bigg)^2, \ \ i=1,2,\dots,q.
\end{equation}

One can see that $(\underbrace{1,\dots,1,h}_m,1,\dots,1)$ is an
invariant line for \eqref{eq11} ($m=1,\dots,q$). On such kind of
invariant line equation \eqref{eq11} reduces to the following fixed
point problem
\begin{equation}\label{eq12}
x=f(x),
\end{equation}
where
\begin{equation}\label{f(x)}
f(x)=\bigg(\frac{\t x+q}{x+\t+q-1}\bigg)^2.
\end{equation}

A simple calculation shows that \eqref{eq12} has a form
$$
(x-1)(x^2+(2\t-\t^2+2q-1)x+q^2)=0.
$$

Hence, $x_0=1$ solution defines a $p$-adic quasi Gibbs measure
$\m_0$.

Now we are interested in finding other solutions of \eqref{eq12},
which means we need to solve the following one
\begin{equation}\label{eq13}
x^2+(2\t-\t^2+2q-1)x+q^2=0.
\end{equation}

Observe that the solution of \eqref{eq13} can be formally written by
\begin{equation}\label{eq14}
x_{1,2}=\frac{-(2\theta-\theta^2+2q-1)\pm(\theta-1)\sqrt{D(\t,q)}}{2},
\end{equation}
where $D(\t,q)=\theta^2-2\theta-4q+1$

So, if the defined solutions exist in $\bq_p$, then they define
$p$-adic quasi Gibbs measures $\m_{1}$ and
$\m_2$, respectively. Note that to exist such solutions the expression $\sqrt{D(\t,q)}$ should have
a sense in $\bq_p$, since in $\bq_p$ not every quadratic equation has a
solution (see Lemma \ref{quadrat}). Therefore, we are going to check
when $\sqrt{D(\t,q)}$ does exist.

Now consider two distinct cases with respect to $N$.

\subsection{Ferromagnetic case}

In this case, we assume that $N>0$, this means $|\t|_p=p^{-N}<1$.
Now let us consider several cases with respect to $q$.

{\sc Case $q=1$}. Note that this case corresponds to the $p$-adic
Ising model, and $D(\t,1)=\theta^2-2\theta-3$.
\begin{enumerate}
\item[(i)] Let $p=2$. Then from $-3=1+2^2+2^3+\cdots$ one has
$$
D(\t,1)=1+2^2+2^3+2^4\e-2\t+\t^2,
$$
where $\e=1+2+2^2+\cdots$. Hence, due to Lemma \ref{quadrat} one can
check that for any $N\geq 1$ the $\sqrt{D(\t,1)}$ does not exist.

\item[(ii)] Let $p=3$. Then taking into account that $\t=p^N$ we find
$$
D(\t,1)=3(3^{2N-1}-2\cdot 3^{N-1}-1).
$$
If $N=1$ then $D(\t,1)=0$, so $\sqrt{D(\t,1)}$ exists. If $N>1$ then
due to Lemma \ref{quadrat} we conclude that $\sqrt{D(\t,1)}$ does
not exist.

\item[(iii)] Let $p\geq 5$. Then from the expression
$$-3=p-3+(p-1)p+(p-1)p^2+\cdots
$$
we obtain
$$
D(\t,1)=p-3+(p-1)p\e_1-2p^N+p^{2N},
$$
where $\e_1=1+p+p^2+\cdots$. So, according to Lemma \ref{quadrat}
$\sqrt{D(\t,1)}$ exists if and only if the equation $x^2\equiv p-3
(\textrm{mod}\ p)$ has a solution in $\bz$. It is easy to see that
the last equation equivalent to $x^2+3\equiv 0(\textrm{mod}\ p)$.
For example, when $p=7$ the equation $x^2+3\equiv 0(\textrm{mod}\
p)$ has a solution $x=2$. So, in this case $\sqrt{D(\t,1)}$ exists.
\end{enumerate}

Hence, we can formulate the following

\begin{thm}\label{q1} Let $N\geq 1$ and $q=1$ (ferromagnetic Ising model). Then the following
assertions hold true:
\begin{enumerate}
\item[(i)] If $p=2$, then there is a unique translation-invariant $p$-adic quasi Gibbs measure
$\m_0$;\\
\item[(ii)] Let $p=3$. If $N=1$, then there are three translation-invariant $p$-adic quasi Gibbs measures $\m_0$, $\m_1$ and $\m_2$, otherwise  there is a unique translation-invariant $p$-adic quasi Gibbs measure $\m_0$;
\item[(iii)] Let $p\geq 5$, then there are three translation-invariant $p$-adic quasi Gibbs measures $\m_0$, $\m_1$ and $\m_2$ if and only if $-3$ is a quadratic residue of modulo $p$, otherwise  there is a unique translation-invariant $p$-adic quasi Gibbs measure $\m_0$;
\end{enumerate}
\end{thm}

{\sc Case $q\geq 2$}. This case corresponds to $q+1$-state Potts
model. Here we shall again consider several cases, but with respect
to $p$.

\begin{enumerate}
\item[(i)] Let $p=2$. Let us represent $q$ in a 2-adic form, i.e.
$$
q=k_0+k_12+\cdots+k_s2^s, \ \ \ s \geq 1.
$$
Then we have
$$
-4q=2^2\big((2-k_0)+(1-k_1)2+\cdots+(1-k_s)2^s\big).
$$
Therefore, one has
$$
 D(\t,q)=1+2^2\big((2-k_0)+(1-k_1)2+\cdots+(1-k_s)2^s\big)-2^{N+1}+2^{2N}.
$$
Now according to Lemma \ref{quadrat} we conclude that
$\sqrt{D(\t,q)}$ exists if and only if $k_0=0$, which is equivalent
to $|q|_2\leq 1/2$.

\item[(ii)] Let $p=3$. We represent $q$ in a 3-adic form, i.e.
$$
q=k_0+k_13+\cdots+k_s3^s, \ \ \ s \geq 0.
$$
Then we have
\begin{eqnarray*}
D(\t,q)&=&1-q-q\cdot 3-2\cdot 3^N+3^{2N}\nonumber \\
&=&1+(3-k_0)+(2-k_1)3+\cdots(2-k_s)3^s-q\cdot 3-2\cdot 3^N+3^{2N}.
\end{eqnarray*}

If $k_0=0$, then from Lemma \ref{quadrat} one can see that
$\sqrt{D(\t,1)}$ exists.

If $k_0=2$, then $\sqrt{D(\t,q)}$ does not exists, since $x^2\equiv
2(\textrm{mod}\ 3)$ has no solution in $\bz$.

 If $k_0=1$, then this case
is more complicated. We cannot provide certain rule to check the
existence of $\sqrt{D(\t,q)}$. But in this case, $\sqrt{D(\t,q)}$
may exist or may not. For example, if $k_1\neq 2$ then
$\sqrt{D(\t,q)}$ does not exist whenever $N\geq 3$. If $k_1=2$ and
$k_2=2$ then $\sqrt{D(\t,q)}$ exists whenever $N\geq 4$.

\item[(iii)] Let $p\geq 5$. Let us  represent $q$ in a $p$-adic expression
$$q=k_0+k_1p+\cdots+k_sp^s, \ \ \ s \geq 0.
$$
Then we have
$$
D(\t,1)=1+4(p-k_0)+4(p-1-k_1)p+\cdots+4(p-1-k_s)p^s-2p^N+p^{2N}.
$$
So, according to Lemma \ref{quadrat} $\sqrt{D(\t,q)}$ exists if the
equation $x^2\equiv 1-4k_0(\textrm{mod}\ p)$ has a solution in $\bz$
whenever $1-4k_0$ is not divided by $p$.  It is clear that if
$k_0=0$ then the equation has a solution for any value of $p$
($p\geq 5$). Note that if $1-4k_0$ is divided by $p$, then
$\sqrt{D(\t,q)}$ does not exist.

If $p=5$  and $k_0=3$, then one can check that $x^2\equiv
-11(\textrm{mod}\ 5)$ has a solution $x=5n+2$. So, in this case
$\sqrt{D(\t,q)}$ exists.
\end{enumerate}

So, we have the following

\begin{thm}\label{q2} Let $N\geq 1$ and $q\geq 2$ (ferromagnetic Potts model). Then the following
assertions hold true:
\begin{enumerate}
\item[(i)] If $|q|_p<1$, then there are three translation-invariant $p$-adic quasi Gibbs measures $\m_0$, $\m_1$ and $\m_2$;
\item[(ii)] Let $p=3$. If $|q-2|_p<1$, then there is a unique translation-invariant $p$-adic quasi Gibbs measure $\m_0$; if $|q-1|_p<1$ there is at least one translation-invariant $p$-adic quasi Gibbs measures $\m_0$;
\item[(iii)] Let $p\geq 5$ and $|4q-1|_p<1$, then there is a unique translation-invariant $p$-adic quasi Gibbs measure $\m_0$;
\end{enumerate}
\end{thm}

\subsection{Antiferromagnetic case}
Now suppose that $N<0$. Denoting $\bar N=-N$, one has $\t=p^{-\bar
N}$. Therefore, $D(\t,q)$ can be represented as follows:
\begin{equation}\label{aD1}
D(\t,q)=p^{-\bar N}\big(1-2p^{\bar N}-4qp^{2\bar N}+p^{2\bar N}\big).
\end{equation}
Hence, due to Lemma \ref{quadrat} we conclude that $\sqrt{D(\t,q)}$
exists for all values of prime $p$, but $\bar N$ should be even.

\begin{thm}\label{aq2} Let $N<0$ and $q\geq 1$ (antiferromagnetic Potts model). If
$-N$ is even, than, for the model \eqref{Potts}, there are three
translation-invariant $p$-adic quasi Gibbs measures $\m_0$, $\m_1$
and $\m_2$.
\end{thm}

\section{Behavior of the dynamical system \eqref{f(x)}}

In this section we are going to investigate the dynamical system
given by \eqref{f(x)}. In the previous section, we have established
some conditions for the existence of its fixed points. In the
sequel, we are going to describe possible attractors of the system,
which allows us to find a relation between behavior of that
dynamical system and the phase transitions. In what follows, for the
sake of simplicity, we always assume that $p\geq 3$.

From \eqref{f(x)} we easily find the following auxiliary facts:
\begin{equation}\label{df(x)}
f'(x)=\bigg(\frac{\t
x+q}{x+\t+q-1}\bigg)^2\cdot\frac{2(\t-1)(\t+q)}{(\t x+q)(x+\t+q-1)};
\end{equation}
\begin{eqnarray}\label{f-xy}
|f(x)-f(y)|_p=\frac{|\t-1|_p|\t+q|_p|x-y|_p|\eta(\t,q;x,y)|_p}{|x+\t+q-1|_p^2|y+\t+q-1|_p^2},
\end{eqnarray}
where
\begin{equation}\label{eta}
\eta(\t,q;x,y)=A\t(x+y)+2\t xy+2qA+q(x+y),
\end{equation}
here $A=\t+q-1$. Furthermore, we assume that $f(x)$ has three fixed
points, the existence such points has been investigated in section
4. We denote them as follows $x_0$, $x_1$, $x_2$. Note that $x_0=1$.
For the fixed points $x_1$ and $x_2$, from \eqref{eq13} we find that
\begin{equation}\label{vieta-h}
x_1+x_2=-2q+1+\t^2-2\t, \ \ \ x_1\cdot x_2=q^2.
\end{equation}

To study dynamics of $f$ we shall consider two different settings
with respect to ferromagnetic and antiferrimagmetic ones.

\subsection{Ferromagnetic case} In this setting, we suppose that $|\t|_p<1$,
moreover $|\t|_p\leq |q|^2$ if $|q|_p<1$.

\begin{lem}\label{fx12-p} Let $x_1$ and $x_2$ be the fixed points of
$f(x)$. Then the followings hold true:
\begin{eqnarray}\label{fx-12-q1}
&&|x_1|_p=|q|^2_p, \ \ |x_2|_p=1, \ \ \textrm{if} \ \ |q|_p<1,\\[2mm]
\label{fx-12-q2} &&|x_1|_p=1, \ \ |x_2|_p=1, \ \ \textrm{if} \ \
|q|_p=1.
\end{eqnarray}
\begin{eqnarray}\label{fx1-q}
&&|\t x_1+q|_p=|q|_p, \ \ |x_1+\t+q-1|_p=1, \ \ \textrm{if} \ \ |q|_p<1,\\[2mm]
\label{fx2-q}
&&|\t x_2+q|_p=|q|_p, \ \ |x_2+\t+q-1|_p=|q|_p, \ \ \textrm{if} \ \ |\t|_p\leq |q|_p^2, |q|_p<1,\\[2mm]
\label{fx12-q} &&|\t x_i+q|_p=1, \ \ |x_i+\t+q-1|_p=1, \ \
\textrm{if} \ \ |q|_p=1, \ i=1,2.
\end{eqnarray}
\end{lem}

\begin{proof}
 First assume that $q$ is divided by $p$, i.e. $|q|_p\leq 1/p$.
Note that, in this case, according to Theorem \ref{q2} there exist
the solutions $x_1$ and $x_2$. Hence, from \eqref{vieta-h} we
conclude that $|x_1+x_2|_p=1$ and $|x_1\cdot x_2|_p=|q^2|_p$. From
the last equalities, without loss of generality, it yields that
\eqref{fx-12-q1}. Hence, we immediately  obtain \eqref{fx1-q}.
 The equality \eqref{vieta-h} implies that
$$
x_2-1=\t^2-2\t-2q-x_1.
$$
This with the strong triangle inequality and \eqref{fx-12-q1} yields
\begin{eqnarray*}
&&|x_2+\t+q-1|_p=|\t^2-\t-q-x_1|_p=|q|_p,\\
&&|\t x_2+q|_p=|q|_p,
\end{eqnarray*}
if $|\t|_p\leq |q|^2_p$.

Now suppose that $|q|_p=1$, and there exist solutions $x_1$ and
$x_2$. Note that, in general, the solutions may not exist (see
Theorems \ref{q1} and \ref{q2}).  Then from \eqref{vieta-h} we find
that
\begin{eqnarray}\label{n-h-12-2}
|x_1+x_2|_p\leq 1,\\
\label{n-h-12-3} |x_1\cdot x_2|_p=1.
\end{eqnarray}
In this case, one has $|x_1|_p=1$, $|x_2|_p=1$. Indeed, assume that
$|x_1|_p<1$, then the equality \eqref{n-h-12-3} yields $|x_2|_p>1$.
Due to the strong triangle inequality we get $|x_1+x_2|_p>1$ which
contradicts to \eqref{n-h-12-2}.

So, due to $|\t|_p<1$,  we have $|\t x_i+q|_p=1$. On the other hand,
we know that $x_i$ ($i=1,2$) are solutions \eqref{eq12}, therefore,
from \eqref{eq12} one gets
\begin{equation*}\label{e-h-12}
|x_i+\t+q-1|_p^2=\frac{|\t x_i+q|_p^2}{|x_i|_p}=1.
\end{equation*}
This completes the proof.
\end{proof}

Let us find behavior of the fixed points. From \eqref{df(x)} we find
\begin{equation}\label{fdf-x0}
|f'(x_0)|_p=\bigg|\frac{\t-1}{\t+q}\bigg|_p= \left\{
\begin{array}{ll}
1, \ \ \ \  \textrm{if}\ |q|_p=1,\\[2mm]
1/|q|_p, \ \ \textrm{if}\ |q|_p<1.
\end{array}
\right.
\end{equation}

Let us consider the other fixed points. Again from \eqref{df(x)} one
gets
\begin{equation}\label{fdf-xi}
|f'(x_i)|_p=\frac{|x_i|_p|\t-1|_p|\t+q|_p}{|\t
x_i+q|_p|x_i+\t+q-1|_p}, \ \ \ i=1,2.
\end{equation}

Now taking into account \eqref{fx-12-q1}--\eqref{fx12-q} we derive
\begin{eqnarray*}
&&|f'(x_1)|_p= \left\{
\begin{array}{ll}
1, \ \ \ \  \textrm{if}\ |q|_p=1,\\[2mm]
|q|_p^2, \ \ \textrm{if}\ |q|_p<1,
\end{array}
\right.\\[2mm]
&&|f'(x_2)|_p= \left\{
\begin{array}{ll}
1, \ \ \ \  \textrm{if}\ |q|_p=1,\\[2mm]
1/|q|_p, \ \ \textrm{if}\ |q|_p<1,
\end{array}
\right.\\[2mm]
\end{eqnarray*}

Consequently, one has

\begin{prop} Let $|\t|_p<1$ and assume that the dynamical system $f$ given by \eqref{f(x)} has three fixed points $x_0$,$x_1$, $x_2$.
Then the following assertions hold true:
\begin{enumerate}
\item[(i)] if $|q|_1=1$, then the fixed points are neutral;

\item[(ii)] if $|q|_p<1$ and $|\t|_p\leq|q|^2_p$, then $x_1$ is attractive, and $x_0$,$x_2$ are repelling.
\end{enumerate}
\end{prop}

Furthermore, we concentrate ourselves to the case $|q|_p<1$, which
is more interesting.

For a given set $B\subset\bq_p$, let us denote
\begin{equation}\label{JB}
J(B)=\{x\in S_1(0) :\ f^n(x)\in B \ \ \textrm{for some} \ n\geq 0\}.
\end{equation}

\begin{thm}\label{f-att} Let $|q|_p<1$, and $|\t|_p\leq|q|^2_p$. Then one has
$$
A(x_1)\supset \{x\in\bq_p:\ |x|_p\neq 1\}\cup \{x\in S_1(0):\
|x-1|_p>|q|_p\}\cup J(B_{|q|^2_p,|q|_p}(x_0))\cup
J(B_{|q|^2_p,|q|_p}(x_2))
$$
\end{thm}

\begin{proof} Let us consider several cases with respect to $|x|_p$.

(I) Assume that $x\in B_1(0)$, then one finds $|f(x)|_p=|q|_p^2<1$,
hence $f(B_1(0))\subset B_1(0)$.

Note that in the considered case we have $|A|_p=1$, therefore for
$x\in B_1(0)$ from \eqref{eta} one immediately gets
$|\eta(\t,q;x,x_1)|_p=|q|_p$. So, \eqref{f-xy},\eqref{fx1-q} with
$|x+\t+q-1|_p=1$ imply that
$$
|f(x)-x_1|_p=|q|_p^2|x-x_1|_p.
$$
Hence, $f$ is a contraction of $B_1(0)$, which means $f^n(x)\to x_1$
for every $x\in B_1(0)$, i.e. $B_1(0)\subset A(x_1)$.

Note that $\bar B_1(0)\varsubsetneq A(x_1)$, since
$|x_0|_p=|x_2|_p=1$, i.e. $S_1(0)\varsubsetneq A(x_1)$.

(II) Assume that $1<|x|_p\leq \frac{|q|_p}{|\t|_p}$, then $|\t
x+q|_p\leq |q|_p$, therefore one finds
$$
|f(x)|_p=\bigg|\frac{\t
x+q}{x+\t+q-1}\bigg|^2\leq\bigg(\frac{|q|_p}{|x|_p}\bigg)^2\leq
|q|^2<1.
$$

(III) Now let $|x|_p>\frac{|q|_p}{|\t|_p}$, then $|\t x+q|_p=|\t
x|_p$, so we have
$$
|f(x)|_p=\frac{|\t x|^2_p}{|x|^2_p}=|\t|^2<1.
$$

Hence, from (II), (III) one concludes that $f(x)\in B_1(0)$, for any
$x$ with $|x|_p>1$, which, due to (I), yields $x\in A(x_1)$.

Consequently, we infer that
\begin{equation}\label{f-A1}
\{x\in\bq_p:\ |x|_p\neq 1\}\subset A(x_1).
\end{equation}

(IV) Now assume that $|x|_p=1$, $|x-1|_p>|q|_p$. Then
$|x+\t+q-1|_p=|x-1|_p$, so one finds
$$
|f(x)|_p=\frac{|q|^2_p}{|x-1|^2_p}<1,
$$
which, due to (I), implies $x\in A(x_1)$.

(V) Suppose that $|x-1|_p<|q|_p$. Then $|x+\t+q-1|_p=|q|_p$, and
from \eqref{eta} we find $|\eta(\t,q;x,1)|_p=|q|_p$. Consequently,
\eqref{f-xy} implies
\begin{equation}\label{f-x-1}
|f(x)-1|_p=\frac{|x-1|_p}{|q|_p}.
\end{equation}
Hence, if $|x-1|_p>|q|^2_p$, then $|f(x)-1|_p>|q|_p$, which, due
(IV), means $x\in A(x_1)$.

(VI) Consider $J(B_{|q|_p^2,|q|_p}(x_0)$. One can see that
$J(B_{|q|_p^2,|q|_p}(x_0)\subset A(x_1)$. Indeed, if $x\in
J(B_{|q|_p^2,|q|_p}(x_0)$, then $|q|^2<|f^{n_0}(x)-1|_p<|q|_p$ for
some $n_0\in\bn$. From \eqref{f-x-1} we obtain
$|f^{n_0+1}(x)-1|_p>|q|_p$, which with (V) yields $x\in A(x_1)$.

Now look to $x_2$. From \eqref{vieta-h} one finds
\begin{equation}\label{fx2-p}
|x_2-1|_p=|q|_p, \ \ |x_2-1+q|_p=|q|_p.
\end{equation}
Note that the strong triangle inequality implies that
$|x-x_2|_p>|q|_p$ if and only if $|x-1|_p>|q|_p$.

(VII) Therefore, assume that $|x-x_2|_p<|q|_p$, which implies that
$|x-1|_p=|q|_p$. So, by means of \eqref{fx2-q},\eqref{fx2-p} from
\eqref{eta} we derive that $|\eta(\t,q;x,x_2)|_p=|q|^2_p$. Hence,
from \eqref{f-xy} with \eqref{fx2-q} and $|x+\t-1+q|_p=|q|_p$ one
finds
\begin{equation}\label{f-x-2}
|f(x)-x_2|_p=\frac{|x-x_2|_p}{|q|_p}.
\end{equation}

Now using the same argument as in (V)-(VII) with \eqref{f-x-2} we
obtain that $J(B_{|q|_p^2,|q|_p}(x_2))\subset A(x_1)$. Note that the
sets $J_1$ and $J_2$ are disjoint. This completes the proof.
\end{proof}

Now we are going to investigate solutions of \eqref{eq1} over the
invariant line $(1,1,\dots,h,1,\dots,1)$. Let us introduce some
notations. If $x\in W_n$, then instead of $h_x$ we use the symbol
$h_x^{(n)}$.

Denote
\begin{equation}\label{g(x)}
g(x)=\frac{\t x+q}{x+\t+q-1}.
\end{equation}
Note that $f(x)=(g(x))^2$. Then one can see that
\begin{eqnarray}\label{g-xy}
&&|g(x)-g(x)|_p=\frac{|x-y|_p|\t-1|_p|\t+q|_p}{|x+\t+q-1|_p|y+\t+q-1|_p},\\[3mm]
\label{g-1} && g^{-1}(x)=\frac{(\t+q-1)x-q}{\t-x}
\end{eqnarray}
Moreover, one has the following
\begin{lem}\label{g-p} Let $|q|_p<1$, and $|\t|_p\leq|q|^2_p$. The following assertions hold true:
\begin{enumerate}
\item[(i)] If $|x|_p\neq 1$, then $|g(x)|_p\leq\max\{|q|_p,|\t|_p\}$;
\item[(ii)] If $|g(x)|_p>1$, then $|x|_p=1$.
\end{enumerate}
\end{lem}
\begin{proof} (i). Let $|x|_p<1$, then from \eqref{g(x)} we get
$$
|g(x)|_p=\bigg|\frac{\t x+q}{x+\t+q-1}\bigg|_p=|q|_p<1.
$$
Now assume $|x|_p>1$, then analogously one finds
$$
|g(x)|_p\left\{
\begin{array}{ll}
=|\t|_p, \ \ \textrm{if} \ |x|_p>\frac{|q|_p}{|\t|_p},\\[2mm]
\leq |q|_p, \ \ \textrm{if} \ 1<|x|_p\leq \frac{|q|_p}{|\t|_p}.
\end{array}
\right.
$$

(ii) Denoting $y=g(x)$, from \eqref{g-1} one finds
$$
|x|_p=|g^{-1}(y)|_p=\bigg|\frac{(\t+q-1)y-q}{\t-y}\bigg|_p=\frac{|y|_p}{|y|_p}=1.
$$
\end{proof}

\begin{thm}\label{U-1} Let $|q|_p<1$, and $|\t|_p\leq|q|^2_p$.
Assume that $\{h_x\}_{x\in V\setminus\{(0)\}}$ is a solution of \eqref{eq1} such that $|h_x|_p\neq 1$ for all $x\in V\setminus\{(0)\}$. Then $h_x=x_1$ for every $x$.
\end{thm}
\begin{proof} Let us first show that  $|h_x|_p<1$ for all $x$.
Suppose that $|h_x^{(n_0)}|_p>1$ for some $n_0\in\bn$ and $x\in
W_{n_0}$. Since $\{h_x\}$ is a solution of \eqref{eq1}, therefore,
we have
\begin{equation}\label{h-x}
h_x^{(n_0)}=g(h_{(x,1)}^{(n_0+1)})g(h_{(x,2)}^{(n_0+1)}),
\end{equation}
here we have used coordinate structure of the tree.

Now according to $|h_{(x,1)}^{(n_0+1)}|_p\neq 1$,
$|h_{(x,2)}^{(n_0+1)}|_p\neq 1$, then Lemma \ref{g-p} (i) implies
that $|g(h_{(x,1)}^{(n_0+1)})|_p<1$, $|g(h_{(x,1)}^{(n_0+1)})|_p<1$,
which with \eqref{h-x} means $|h_{x}^{(n_0)}|_p<1$. It is a
contradiction.

Hence, $|h_x|_p<1$ for all $x$. Then from \eqref{g-xy} we obtain
\begin{equation}\label{g-xy-11}
|g(h_{x})-g(x_1)|_p=|q|_p|h_x-x_1|_p
\end{equation}
for any $x\in V\setminus\{(0)\}$.

Now denote
$$
\|h^{(n)}\|_p=\max\{|h_{x}^{(n)}|_p:\ x\in W_n\}.
$$

Let $\e>0$ be an arbitrary number. Then from the prof of Lemma
\ref{g-p} (i) with \eqref{g-xy-11} one finds
\begin{eqnarray}\label{f-hx-x11}
|h_x^{(n)}-x_1|_p&=&|g(h_{(x,1)}^{(n+1)})g(h_{(x,2)}^{(n+1)})-(g(x_1))^2|_p\nonumber\\
&=&\bigg|g(h_{(x,1)}^{(n+1)})\big(g(h_{(x,2)}^{(n+1)})-g(x_1)\big)+g(x_1)\big(g(h_{(x,2)}^{(n+1)})-g(x_1)\big)\bigg|_p\nonumber\\
&\leq&\max\bigg\{|g(h_{(x,1)}^{(n+1)})|_p\big|g(h_{(x,2)}^{(n+1)})-g(x_1)\big|_p,
|g(x_1)|_p\big|g(h_{(x,2)}^{(n+1)})-g(x_1)\big|_p\bigg\}\\
&\leq &
|q|^2_p\max\bigg\{|h_{(x,1)}^{(n+1)}-x_1|_p,|h_{(x,2)}^{(n+1)}-x_1|_p\bigg\}.
\nonumber
\end{eqnarray}
Thus, we derive
$$
\|h^{(n)}-x_1\|_p\leq |q|^2_p\|h^{(n+1)}-x_1\|_p.
$$
So, iterating the last inequality $N$ times one gets
\begin{equation}\label{f-iter1}
\|h^{(n)}-x_1\|_p\leq |q|^{2^N}_p\|h^{(n+N)}-x_1\|_p.
\end{equation}
Choosing $N$ such that $|q|^{2^N}_p<\e$, from \eqref{f-iter1} we
find $\|h^{(n)}-x_1\|_p<\e$. Arbitrariness of $\e$ yields that
$h_x=x_1$. This completes the proof.
\end{proof}

\subsection{Antiferromagnetic case} In this subsection we assume that $N<0$, this means $|\t|_p=p^{\bar N}>1$, where
$\bar N=-N$. In this setting equation \eqref{eq12} has three
solutions $x_0$ (i.e. $x_0=1$), and $x_1,x_2$. Note that $x_1$ and
$x_2$ are solutions of \eqref{eq13}, therefore one gets
\begin{eqnarray}\label{ax-12-q1}
&&|x_1+x_2|_p=|\t|_p^2, \ \ \ |x_1\cdot x_2|_p=|q|_p^2\\
\label{ax-12-q2} &&|x_1|_p=|\t|^2_p, \ \
|x_2|_p=\bigg|\frac{q}{\t}\bigg|_p^2.
\end{eqnarray}
Hence, we obtain
\begin{eqnarray}\label{ax1-q}
&&|\t x_1+q|_p=|\t|^3_p, \ \ |x_1+\t+q-1|_p=|\t|^2_p, \\[2mm]
\label{ax2-q} &&|\t x_2+q|_p=|q|_p, \ \ |x_2+\t+q-1|_p=|\t|_p.
\end{eqnarray}

\begin{prop} Assume that $|\t|_p>1$, then a fixed point $x_0$ is neutral, and the fixed points $x_1$, $x_2$ are
attractive.
\end{prop}

\begin{proof} From \eqref{df(x)} we find
\begin{equation*}
|f'(x_0)|_p=\bigg|\frac{\t-1}{\t+q}\bigg|_p=1,
\end{equation*}
this means that $x_0$ is neutral.

Let us consider the other fixed points. Now taking into account
\eqref{fdf-xi} with \eqref{ax-12-q1},\eqref{ax1-q},\eqref{ax2-q} one
gets
\begin{eqnarray*}
|f'(x_1)|_p=\frac{1}{|\t|_p}<1, \ \ \
|f'(x_2)|_p=\bigg|\frac{q}{\t}\bigg|_p<1
\end{eqnarray*}
which is the required assertion.
\end{proof}

\begin{lem}\label{af-pp} Let $|\t|_p>1$ and $f$ be given by \eqref{f(x)}. Then the following assertions
hold true:
\begin{enumerate}
\item[(i)] if $|x|_p>|\t|_p$, then $|f(x)|_p=|\t|_p^2$. Hence, $|f^n(x)|_p=|\t|_p^2$ for all $n\in\bn$;
\item[(ii)] if $\frac{|q|_p}{|\t|_p}< |x|_p<|\t|_p$, then $|f(x)|_p=|x|^2_p$;
\item[(iii)] if $|x|\leq \frac{|q|_p}{|\t|_p}$, then $|f(x)|_p\leq\big(\frac{|q|_p}{|\t|_p}\big)^2$.
\end{enumerate}
\end{lem}

\begin{proof} (i). Let  $|x|_p>|\t|_p$, then from \eqref{f(x)} we find
\begin{equation*}
|f(x)|_p=\bigg(\frac{|\t x|_p}{|x|_p}\bigg)^2=|\t|_p^2.
\end{equation*}

(ii) Let $\frac{|q|_p}{|\t|_p}< |x|_p<|\t|_p$, then
$|\t+x+q-1|_p=|\t|_p$, $|\t x|_p>|q|_p$ therefore, one gets
\begin{equation*}
|f(x)|_p=\bigg(\frac{|\t x|_p}{|\t|_p}\bigg)^2=|x|_p^2.
\end{equation*}

(iii) Let $|x|\leq\frac{|q|_p}{|\t|_p}$, then $|\t x|_p\leq|q|_p$,
$|\t+x+q-1|_p=|\t|_p$, hence one finds the required equality.
\end{proof}

Now we are going to examine attractors of $x_1$ and $x_2$.

\begin{thm} Let $|\t|_p>1$ and $f(x)$ is given by \eqref{f(x)}. Then the following assertions holds true:
\begin{enumerate}
\item[(i)] $f(S_1(0))\subset S_1(0)$;
\item[(ii)] $A(x_2)=B_1(0)$;
\item[(iii)] $A(x_1)=\{x\in\bq_p: \ |x|_p>1\}\setminus\bigcup\limits_{n=0}^\infty f^{-n}(1-\t-q)
$.
\end{enumerate}
\end{thm}

\begin{proof} (i) Let $|x|_p=1$, then due to Lemma \ref{af-pp} (ii) we find that $|f(x)|_1=1$, which means $S_1(0)$ is invariant w.r.t. $f$.\\

Now consider (ii). First note that  $|x_2|_p=(|q|_p/|\t|_p)^2<1$.
Now consider several cases w.r.t. $|x|_p$.

(I$_1$) Assume that $|x-x_2|_p<\big(\frac{|q|_p}{|\t|_p}\big)^2$.
Then $|x+\t+q-1|_p=|\t|_p$ and
$|x+x_2|_p=\big(\frac{|q|_p}{|\t|_p}\big)^2$. Therefore, from
\eqref{eta} one gets that
$$
|\eta(\t,q;x,x_2)|_p=|\t|^2_p\bigg(\frac{|q|_p}{|\t|_p}\bigg)^2=|q|_p^2.
$$
Hence, the last equality with \eqref{f-xy},\eqref{ax2-q} yields that
\begin{eqnarray}\label{af-xx2}
|f(x)-x_2|_p&=&\frac{|\t|_p^2|x-x_2|_p|\eta(\t,q;x,x_2)|_p}{|x+\t+q-1|_p^2|\t|_p^2}\\[2mm]
&=&\frac{|q|_p^2}{|\t|_p^2}|x-x_2|_p \nonumber.
\end{eqnarray}
This means that $f$ maps $B_{\frac{|q|^2_p}{|\t|_p^2}}(x_2)$ into
itself, and it is a contraction. So, for every $x\in
B_{\frac{|q|^2_p}{|\t|_p^2}}(x_2)$, one has $f^n(x)\to x_2$ as
$n\to\infty$. Hence, $B_{\frac{|q|^2_p}{|\t|_p^2}}(x_2)\subset
A(x_2)$.

(II$_1$) Let  $|x-x_2|_p=\frac{|q|^2}{|\t|_p^2}$, i.e. $|x|_p\leq
\frac{|q|^2}{|\t|_p^2}$. Then from \eqref{eta} we find that
$|\eta(\t,q;x,x_2)|_p=|q|_p|\t|_p$. Hence, from \eqref{af-xx2} we
derive
\begin{equation*}
|f(x)-x_2|_p=\frac{|x-x_2|_p||q|_p|\t|_p}{|\t|_p^2}=\frac{|q|^3_p}{|\t|^3_p},
\end{equation*}
which implies that $f(x)\in B_{\frac{|q|^2_p}{|\t|_p^2}}(x_2)$,
hence due to (I$_1$) one has $x\in A(x_2)$.

(III$_1$) Let $|x_2|_p<|x|_p\leq \frac{|q|_p}{|\t|_p}$, then
$|\eta(\t,q;x,x_2)|_p\leq |q|_p|\t|_p$, therefore using the same
argument as (II$_1$) we obtain
\begin{equation*}
|f(x)-x_2|_p\leq \frac{|q|^3_p}{|\t|^3_p},
\end{equation*}
which with (I$_1$) yields $x\in A(x_2)$.

(IV$_1$) Let $\frac{|q|_p}{|\t|_p}<|x|_p<1$, then  $|x-x_2|_p=|x|_p$
and $|\eta(\t,q;x,x_2)|_p=|\t|^2_p|x|_p$. It follows from
\eqref{af-xx2} that
\begin{equation}\label{af-xxx}
|f(x)-x_2|_p=\frac{|x-x_2|_p|\t|^2_p|x|_p}{|\t|_p^2}=|x|^2_p.
\end{equation}
If $|x|^2_p\leq \frac{|q|_p}{|\t|_p}$, then $f(x)$ falls to
(III$_1$) case, so $x\in A(x_2)$. If
 $|x|^2_p>\frac{|q|_p}{|\t|_p}$, then again repeating \eqref{af-xxx} one gets
$|f^2(x)-x_2|_p=|x|^4_p$. Continuing this procedure, we conclude
that in any case $f(x)$ falls to (III$_1$). Hence, $x\in A(x_2)$.

According (i) $S_1(0)$ is invariant w.r.t. $f$, hence $S_1(0)\cap A(x_2)=\emptyset$. Hence, (ii) is proved. \\

Now let us prove (iii). From \eqref{f-xy} with \eqref{ax1-q} we
easily obtain
\begin{equation}\label{af-xx1}
|f(x)-x_1|_p=\frac{|\t|_p^2|x-x_1|_p|\eta(\t,q;x,x_1)|_p}{|x+\t+q-1|_p^2|\t|_p^4},
\end{equation}
where $\eta(\t,q;x,x_1)$ is defined in \eqref{eta}.

(I$_2$) Let $x\in B_{|\t|_p^2}(x_1)$ (i.e. $|x-x_1|_p<|\t^2|_p$).
Then $|x|_p=|x_1|_p=|\t|^2_p$ and $|x+\t+q-1|_p=|\t|^2_p$.
Therefore, from \eqref{eta} one finds that
$|\eta(\t,q;x,x_1)|_p=|\t|_p^5$. So, the last ones with
\eqref{af-xx1} imply
\begin{equation}\label{af-xx11}
|f(x)-x_1|_p=\frac{|x-x_1|_p}{|\t|_p},
\end{equation}
this means that $f$ is a contraction of $B_{|\t|_p^2}(x_1)$, i.e.
for every $x\in B_{|\t|_p^2}(x_1)$ one has $f^n(x)\to x_1$ as
$n\to\infty$. Hence, $B_{|\t|_p^2}(x_1)\subset A(x_1)$. Note that
$S_{|\t|_p^2}(x_1)\nsubseteq A(x_1)$, since $x_2\in
S_{|\t|_p^2}(x_1)$.

(II$_2$) Let $|x|_p=|\t|^2_p$, which implies $|x-x_1|_p\leq
|\t|_p^2$. Similarly reasoning as (I$_2$) we find
$|\eta(\t,q;x,x_1)|_p=|\t|_p^5$, hence \eqref{af-xx11} holds. So,
$f(x)\in B_{|\t|_p^2}(x_1)$, therefore, due to (I$_2$), we get $x\in
A(x_1)$.

(III$_2$) Let us assume that $|x-x_1|>|\t|^2_p$, then
$|x|_p>|\t|_p^2$. This implies $|x-x_1|_p=|x|_p$. So, we have
$|\eta(\t,q;x,x_1)|_p=|\t|^3_p|x|_p$ and $|x+\t+q-1|_p=|x|_p$, hence
from \eqref{af-xx1} one finds
\begin{equation*}
|f(x)-x_1|_p=\frac{|x|_p|\t|^3_p|x|_p}{|x|_p^2|\t|_p^2}=|\t|_p.
\end{equation*}
This implies that $f(x)\in B_{|\t|_p^2}(x_1)$, which with (I$_2$)
means
$$
\{x\in\bq_p: \ |x|_p>|\t|_p^2\}\subset A(x_1).
$$

(IV$_2$) Let $|x|=|\t|_p$ with $x\neq 1-\t-q$. Denote $\g=x+\t-q+1$,
then $|\g|_p\leq |\t|_p$ and $|\g|_p\neq 0$. From \eqref{f(x)} one
finds
\begin{equation}\label{af-fxx}
|f(x)|_p=\frac{|\t|_p^4}{|\g|_p^2}\geq |\t|^2_p.
\end{equation}
Hence due to (II$_2$) and (III$_2$) we conclude that $x\in A(x_1)$.

(V$_2$) Let $|x|_p>|\t|_p$, then analogously from Lemma
\ref{af-pp}(i) one finds that $|f(x)|_p=|\t|^2_p$, which, due to
(II$_2$), yields that $x\in A(x_1)$.

(VI$_2$) Now assume that $1<|x|<|\t|_p$.  Then
$|x+\t+q-1|_p=|\t|_p$, $|\eta(\t,q;x,x_1)|_p=|\t|^4_p$, hence it
follows from \eqref{af-xx1} that
\begin{equation}\label{af-xx12}
|f(x)-x_1|_p=\frac{|x-x_1|_p||\t|^4_p}{|\t|_p^2|\t|_p^2}=|x-x_1|_p.
\end{equation}
On the other hand, according to Lemma \ref{af-pp} (ii) we have
$|f(x)|_p=|x|_p^2$.

(a) If $|x|^2=|\t|_p$, and $f(x)\neq 1-\t-q$, then due to (IV$_2$)
one gets $x\in A(x_1)$. If $f(x)=1-\t-q$, then $x\notin A(x_1)$.

(b) If $|x|^2>|\t|_p$, then from (V$_2$) we get
$|f^2(x)|_p=|\t|_p^2$, hence $x\in A(x_1)$.

(c) If $|x|^2<|\t|_p$, then repeating above made argument we find
$|f^2(x)-x_1|_p=|x-x_1|_p$ and $|f^2(x)|_p=|x|_p^4$.  Therefore,
continuing above made procedure, we conclude there can occur either
(a) or (b). Hence, $x\in A(x_1)$ if
$x\notin\bigcup\limits_{n=0}^\infty f^{-n}(1-\t-q)$.
\end{proof}

To prove our main result, we need the following auxiliary result.

\begin{lem}\label{ag-p} Let $|\t|_p>1$ and $g(x)$ be a function given by \eqref{g(x)}.
The following assertions hold true:
\begin{enumerate}
\item[(i)] if $|x|_p\leq \frac{1}{|\t|_p}$, then $|g(x)|_p\leq\frac{1}{|\t|_p}$;
\item[(ii)] if $\frac{1}{|\t|_p}<|x|_p<1$, then $|g(x)|_p=|x|_p$;
\item[(iii)] if $x,y\in B_1(0)$ then one has
\begin{eqnarray}\label{ag-xy}
|g(x)-g(x)|_p=|x-y|_p;
\end{eqnarray}
\item[(iv)] if $1<|x|_p<|\t|_p$, then $|g(x)|_p=|x|_p$;
\item[(v)] if $|x|_p\geq |\t|_p$, then $|g(x)|_p\geq |\t|_p$. Moreover, if $|x|_p>|\t|_p$, then $|g(x)|_p=|\t|_p$;
\item[(vi)] if $|x|_p,|y|_p\geq |\t|_p^2$ then one has
\begin{eqnarray}\label{ag-xy1}
|g(x)-g(x)|_p\leq\frac{1}{|\t|_p^2}|x-y|_p;
\end{eqnarray}
\end{enumerate}
\end{lem}

\begin{proof} From \eqref{g(x)} we get
$$
|g(x)|_p=\frac{|\t x+q|_p}{|\t|_p} \left\{
\begin{array}{ll}
\leq \frac{1}{|\t|_p}, \ \ \textrm{if} \ \ |x|_p\leq \frac{1}{|\t|_p},\\[3mm]
=|x|_p, \ \ \textrm{if} \ \ \frac{1}{|\t|_p}<|x|_p<1,
\end{array}
\right.
$$
which proves (i) and (ii).

The assertions (iii),(vi) immediately follows from \eqref{g-xy}.

From \eqref{g(x)} we get
$$
|g(x)|_p=\frac{|\t x|_p}{|\t+q-1+x|_p} \left\{
\begin{array}{ll}
=|x|_p, \ \ \textrm{if} \ \ 1<|x|_p <|\t|_p,\\[3mm]
\geq |\t|_p, \ \ \textrm{if} \ \ |x|_p\geq |\t|_p,
\end{array}
\right.
$$
hence one gets (iv) and (v).
\end{proof}

Now we are going to describe solutions of \eqref{eq1}.

\begin{thm}\label{a-s-x2} Let $|\t|_p>0$ and assume that $\{h_x\}_{x\in V\setminus\{(0)\}}$ is a solution of
\eqref{eq1}. Then the following assertions hold ture:
\begin{enumerate}
\item[(i)] if $|h_x|_p<1$ for all $x\in V\setminus\{(0)\}$, then $h_x=x_2$ for every $x$;
\item[(ii)] if $|h_x|_p>1$ for all $x\in
V\setminus\{(0)\}$,  then $h_x=x_1$ for every $x$.
\end{enumerate}
\end{thm}

\begin{proof} Let us prove (i).  First we establish that $|h_x|_p\leq\frac{1}{|\t|_p}$
for all $x\in V\setminus\{(0)\}$. Indeed, suppose
$|h_x^{(n_0)}|_p>\frac{1}{|\t|_p}$ for some $n_0\in\bn$. From
\eqref{eq1} one has
\begin{equation}\label{ah-x}
|h_x^{(n_0)}|_p=|g(h_{(x,1)}^{(n_0+1)})|_p|g(h_{(x,2)}^{(n_0+1)})|_p.
\end{equation}
Now consider some possible cases.

\begin{enumerate}
\item[(a)] if $|h_{(x,1)}^{(n_0+1)}|_p\leq \frac{1}{|\t|_p}$, but
$|h_{(x,2)}^{(n_0+1)}|_p>\frac{1}{|\t|_p}$, then due Lemma
\ref{ag-p}(i),(ii) from \eqref{ah-x} we derive

\begin{eqnarray}\label{ah-1}
\frac{1}{|\t|_p}<|h_x^{(n_0)}|_p&=&|g(h_{(x,1)}^{(n_0+1)})|_p|g(h_{(x,2)}^{(n_0+1)})|_p\nonumber\\[2mm]
&\leq&\frac{1}{|\t|_p}|h_{(x,2)}^{(n_0+1)}|_p\nonumber\\
&<&\frac{1}{|\t|_p} \ \ (\textrm{since} \
|h_{(x,2)}^{(n_0+1)}|_p<1),
\end{eqnarray}
but it is a contradiction.

\item[(b)] if $|h_{(x,2)}^{(n_0+1)}|_p\leq \frac{1}{|\t|_p}$ and
$|h_{(x,1)}^{(n_0+1)}|_p>\frac{1}{|\t|_p}$ similarly as (a) we come
to the contradiction.

\item[(c)] if $|h_{(x,1)}^{(n_0+1)}|_p\leq \frac{1}{|\t|_p}$, and
$|h_{(x,2)}^{(n_0+1)}|_p\leq \frac{1}{|\t|_p}$, then again by the
same argument one finds a contradiction.
\end{enumerate}
Hence, one has $|h_{(x,1)}^{(n_0+1)}|_p>\frac{1}{|\t|_p}$,
$|h_{(x,2)}^{(n_0+1)}|_p> \frac{1}{|\t|_p}$. Therefore, assume that
$|h_{(x,i_1\cdots i_k)}^{(n_0+k)}|_p>\frac{1}{|\t|_p}$ for every
$k\geq 1$. Then, according to \eqref{eq1} with Lemma \ref{ag-p}(ii)
one gets
\begin{eqnarray}\label{ah-2}
|h_{x}^{(n_0)}|_p&=&\prod_{i_1=1,2}|h_{(x,i_1)}^{(n_0+1)}|_p\nonumber\\
&=&\prod_{i_1=1,2}\prod_{i_2=1,2}|h_{(x,i_1i_2)}^{(n_0+2)}|_p\nonumber\\[2mm]
&&\cdots\nonumber\\[2mm]
&=&\prod_{i_1,\dots i_k=1,2}|h_{(x,i_1\cdots i_k)}^{(n_0+k)}|_p.
\end{eqnarray}
Denote
$$
\gamma_k=\max_{i_1,\dots i_k}|h_{(x,i_1\cdots i_k)}^{(n_0+k)}|_p.
$$
We know that $|h_{(x,i_1\cdots i_k)}^{(n_0+k)}|_p<1$ for all
$i_1,\dots i_k$, $k\geq 1$. Therefore, due to our assumption one has
$\frac{1}{|\t|_p}<|\g_k|_p<1$ for every $k\in\bn$. Hence, from
\eqref{ah-2} one finds
\begin{equation}\label{ah-3}
|h_{x}^{(n_0)}|_p\leq |\gamma_k|_p^{2^k}.
\end{equation}
Thus, when $k$ is large enough, then
$|\gamma_k|_p^{2^k}<\frac{1}{|\t|_p}$. So, from \eqref{ah-3} we
obtain $|h_{x}^{(n_0)}|_p<\frac{1}{|\t|_p}$, which is a
contradiction.\\

Take an arbitrary $\e>0$. Then similarly to \eqref{f-hx-x11} one has
\begin{equation}\label{a-hx-x12}
|h_x^{(n)}-x_2|_p\leq\max\bigg\{|g(h_{(x,1)}^{(n+1)})|_p\big|g(h_{(x,2)}^{(n+1)})-g(x_1)\big|_p,
|g(x_1)|_p\big|g(h_{(x,2)}^{(n+1)})-g(x_1)\big|_p\bigg\}.
\end{equation}
According to $|h_x|_p\leq\frac{1}{|\t|_p}$, for every $x$, with
Lemma \ref{ag-p} (i),(iii) from \eqref{a-hx-x12} we derive
\begin{eqnarray*}
\|h^{(n)}-x_2\|_p\leq \frac{1}{|\t|_p}\|h^{(n+1)}-x_2\|_p.
\end{eqnarray*}
So, iterating the last inequality $M$ times one gets
\begin{equation}\label{a-iter1}
\|h^{(n)}-x_2\|_p\leq \frac{1}{|\t|^{M}_p}\|h^{(n+M)}-x_1\|_p.
\end{equation}
Choosing $M$ such that $|\t|^{-M}_p<\e$, from \eqref{a-iter1} we
find $\|h^{(n)}-x_2\|_p<\e$. Arbitrariness of $\e$ yields that
$h_x=x_2$. \\

Now consider (ii).  Let us show that $|h_x|_p\geq |\t|_p$ for all
$x\in V\setminus\{(0)\}$. Assume that from the contrary,
$|h_x^{(n_0)}|_p<|\t|_p$ for some $n_0\in\bn$ and $x$. Then we have
\eqref{ah-x}. Therefore,  consider several possible cases.

\begin{enumerate}
\item[(a)] if $|h_{(x,1)}^{(n_0+1)}|_p\geq |\t|_p$, and $|h_{(x,2)}^{(n_0+1)}|_p<|\t|_p$,
then due Lemma \ref{ag-p}(iv),(v) from \eqref{ah-x} one finds
\begin{eqnarray*}
|\t|_p<|h_x^{(n_0)}|_p&=&|g(h_{(x,1)}^{(n_0+1)})|_p|g(h_{(x,2)}^{(n_0+1)})|_p\nonumber\\[2mm]
&\leq&|\t|_p|h_{(x,2)}^{(n_0+1)}|_p\nonumber\\
&<&|\t|_p \ \ (\textrm{since} \ |h_{(x,2)}^{(n_0+1)}|_p>1),
\end{eqnarray*}
but this is a contradiction.

\item[(b)] if either $|h_{(x,2)}^{(n_0+1)}|_p\geq |\t|_p$,
$|h_{(x,1)}^{(n_0+1)}|_p<|\t|_p$ or $|h_{(x,1)}^{(n_0+1)}|_p\geq
\frac{1}{|\t|_p}$, $|h_{(x,2)}^{(n_0+1)}|_p\geq \frac{1}{|\t|_p}$,
then by the same argument as (a) one finds a contradiction.
\end{enumerate}
Hence, we conclude that $|h_{(x,1)}^{(n_0+1)}|_p<|\t|_p$,
$|h_{(x,2)}^{(n_0+1)}|_p<|\t|_p$. Therefore, assume that
$|h_{(x,i_1\cdots i_k)}^{(n_0+k)}|_p>\frac{1}{|\t|_p}$ for every
$k\geq 1$. Then, according to \eqref{eq1} with Lemma \ref{ag-p}(iv)
one gets
\begin{eqnarray}\label{ah-4}
|h_{x}^{(n_0)}|_p=\prod_{i_1,\dots i_k=1,2}|_ph_{(x,i_1\cdots
i_k)}^{(n_0+k)}|_p.
\end{eqnarray}
Denote
$$
\delta_k=\min_{i_1,\dots i_k}|h_{(x,i_1\cdots i_k)}^{(n_0+k)}|_p.
$$
We know that $|h_{(x,i_1\cdots i_k)}^{(n_0+k)}|_p>1$ for all
$i_1,\dots i_k$, $k\geq 1$. Therefore, from our assumption one has
$1<|\delta_k|_p<|\t|_p$ for every $k\in\bn$. Hence, from
\eqref{ah-4} one finds
\begin{equation}\label{ah-5}
|h_{x}^{(n_0)}|_p\geq |\delta_k|_p^{2^k}.
\end{equation}
Thus, when $k$ is large enough, then $|\delta_k|_p^{2^k}\geq
|\t|_p$. So, from \eqref{ah-5} we obtain $|h_{x}^{(n_0)}|_p\geq
|\t|_p$, which contradicts to $|h_{x}^{(n_0)}|_p<|\t|_p$.

Thus, $|h_{x}|_p\geq|\t|_p$ for all $x$. Then from \eqref{eq1} one
concludes (see also \eqref{ah-x}) that $|h_x|_p\geq |\t|_p^2$. Then
taking an arbitrary $\e>0$ and using \eqref{a-hx-x12} with Lemma
\ref{ag-p} (iv),(vi) we obtain
\begin{eqnarray*}
\|h^{(n)}-x_1\|_p\leq \frac{1}{|\t|_p}\|h^{(n+1)}-x_1\|_p.
\end{eqnarray*}
Now the same argument as (i) we get the desired assertion. This
completes the proof.
\end{proof}

\section{Boundedness of $p$-adic quasi Gibbs measures and phase transitions}

From the results of the previous section, we conclude that to
investigate the quasi $p$-adic measure, for us it is enough to study
the measures $\m_0$,$\m_1$ and $\m_2$, corresponding to the
solutions $x_0$,$x_1$ and $x_2$. In this section we shall study
boundedness and unboundedness of the said measures.

Furthermore, we are going to consider the $p$-adic quasi Gibbs
measures corresponding to these solutions. Due to Lemma \ref{parti}
the partition function $Z_{i,n}$ corresponding to the measure $\m_i$
($i=1,2$) has the following form
\begin{equation}\label{Z-IN}
Z_{i,n}=a_i^{|V_{n-1}|}
\end{equation}
where $a_i=(x_i+\t+q-1)^2h_0$.

For a given configuration $\s\in\Om_{V_n}$ denote
$$
\#\s=\{x\in W_n:\ \s(x)=1\}.
$$

From \eqref{mu},\eqref{H} and \eqref{Z-IN} we find
\begin{eqnarray}\label{e-mu1}
|\m_1(\s)|_p&=&\frac{1}{Z_{1,n}}\cdot
\frac{1}{p^{H(\s)}}\prod_{x\in W_n}\bigg|\frac{h_{\s(x),x}}{h_0}\bigg|_p|h_0|_p^{|W_n|}\nonumber\\
&=&\frac{|h_0|_p^{|W_n|-|V_{n-1}|}}{|x_1+\t+q-1|_p^{2|V_{n-1}|}}\cdot\frac{|x_1|_p^{\#\s}}
{p^{H(\s)}}\nonumber\\
&=&\frac{|h_0|_p^2}{|x_1+\t+q-1|_p^{2|V_{n-1}|}}\cdot\frac{|x_1|_p^{\#\s}}{p^{H(\s)}},
\end{eqnarray}
where we have used the equality $|W_n|-|V_{n-1}|=2$.

Similarly, one gets
\begin{eqnarray}\label{e-mu2}
|\m_2(\s)|_p=\frac{|h_0|_p^2}{|x_2+\t+q-1|_p^{2|V_{n-1}|}}\cdot\frac{|x_2|_p^{\#\s}}{p^{H(\s)}},
\end{eqnarray}

\subsection{Ferromagnetic case} Assume that $N>0$, i.e. $|\t|_p<1$. In this subsection we shall prove the existence of phase
transitions. Namely one has the following

\begin{thm}\label{bound1}
Assume that $|q|_p<1$, $|\t|_p\leq |q|^2_p$.  Then for $p$-adic
quasi Gibbs measures $\m_0$,$\m_1$ $\m_2$ of the ferromagnetic
$q+1$-state Potts model \eqref{Potts} one has: the measure $\m_1$ is
bounded; the measures $\m_0$ and $\m_2$ are unbounded. Moreover,
there is a strong phase transition.
\end{thm}

\begin{proof} According to
Theorem \ref{q2} the conditions $|q|_p<1$, $|\t|_p\leq |q|^2_p$
imply the existence of three translation-invariant measures
$\m_0$,$\m_1$ and $\m_2$.

Then from \eqref{e-mu1} with \eqref{fx-12-q1},\eqref{fx1-q} we
obtain
\begin{equation}\label{e-mu1-2}
|\m_1(\s)|_p=\frac{|h_0|_p^2}{p^{H(\s)}}\cdot |x_1|_p^{\#\s}\leq
|h_0|_p^2,
\end{equation}
which implies that the measure $\m_1$ is bounded.

Similarly, from \eqref{e-mu2} with \eqref{fx-12-q1},\eqref{fx2-q} we
find
\begin{eqnarray}\label{e-mu2-2}
|\m_2(\s)|_p&=&\frac{|h_0|_p^2}{|q|_p^{2|V_{n-1}|}}\cdot\frac{1}{p^{H(\s)}}\nonumber\\
&\geq& |h_0|_p^2p^{2|V_{n-1}|-H(\s)}
\end{eqnarray}

Now let us choose $\s_{0,n}\in\Om_{V_n}$ as follows $\s_{0,n}(x)=1$
for all $x\in V_n$. Then one can see that $H(\s_{0,n})=0$, therefore
it follows from \eqref{e-mu2-2} that
$$
|\m_2(\s_{0,n})|_p\geq |h_0|_p^2p^{2|V_{n-1}|}\to\infty \ \ \textrm{as} \ \ n\to\infty.
$$
This yields that the measure $\m_2$ is not bounded.

Let us consider the measure $\m_0$. Similarly, we obtain
\begin{eqnarray}\label{e-mu0}
|\m_0(\s)|_p&=&\frac{|h_0|_p^2}{|\t+q|_p^{2|V_{n-1}|}}\cdot\frac{1}{p^{H(\s)}}\nonumber\\
&=&\frac{|h_0|_p^2}{|q|_p^{2|V_{n-1}|}}\cdot\frac{1}{p^{H(\s)}}\nonumber\\
&\geq& |h_0|_p^2p^{2|V_{n-1}|-H(\s)}
\end{eqnarray}
so, we immediately find that $|\m_0(\s_{0,n})|_p\to\infty$ as $n\to\infty$.

It follows from \eqref{e-mu2-2}, \eqref{e-mu0} that
$$
\bigg|\frac{\m_0(\s)}{\m_2(\s)}\bigg|_p=1.
$$

Now let us compare $\m_1$ and $\m_2$. From
\eqref{e-mu1-2},\eqref{e-mu2-2} with \eqref{fx-12-q1} one finds
\begin{eqnarray}\label{e-mu12}
|\m_1(\s_{0,n})\m_2(\s_{0,n})|_p&=&\frac{|h_0|_p^4|x_1|_p^{\#\s_{0,n}}}{|q|_p^{2|V_{n-1}|}}\nonumber\\
&=&|h_0|_p^4|q|_p^{2(|W_n|-|V_{n-1}|)}\nonumber\\
&=&|h_0|_p^4|q|_p^4.
\end{eqnarray}
This implies that $|\m_1(\s_{0,n})|_p\to 0$ as $n\to\infty$.
\end{proof}

Now assume that $|q|_p=1$. In this case, the solutions $x_1$ and
$x_2$ may not exists (see Theorems \ref{q1} and \ref{q2}).
Therefore, we suppose the existence of such solutions.

Now taking into account \eqref{e-mu1}, \eqref{e-mu2} with
\eqref{fx-12-q2} we derive
\begin{equation}\label{e-m12}
|\m_i(\s)|_p=\frac{|h_0|_p^2|x_i|_p^{\#\s^{(i)}}}{p^{H(\s)}}=
\frac{|h_0|_p^2}{p^{H(\s)}}\leq |h_0|^2_p \ \ \ (i=1,2),
\end{equation}
so the measures $\m_1$ and $\m_2$ are bounded.

We would like to compare these measure. Therefore, let us consider
the following difference
\begin{eqnarray}\label{e-dif-12}
|\m_0(\s)-\m_i(\s)|_p=\frac{|h_0|_p^2}{p^{H(\s)}}\bigg|(\t+q-1+x_i)^{2|V_{n-1}|}-x_i^{\#\s}(\t+q)^{2|V_{n-1}|}\bigg|_p.
\end{eqnarray}
Denoting
$$
x=\t+q-1, \ \ y=x_i, \ \ N=2|V_{n-1}|, \ \ k=\#\s
$$
and taking into account $|x|_p\leq 1$ and $|y|_p=1$, the right-hand
side of \eqref{e-dif-12} can be estimated as follows
\begin{eqnarray}\label{e-dif-1}
|(x+y)^N-y^k(x+1)^N|_p&=&\bigg|\sum_{\ell=0}^NC_N^\ell x^\ell(y^{N-\ell}-y^k)\bigg|_p\nonumber\\
&=&\bigg|\sum_{\ell=0}^NC_N^\ell x^\ell y^{\min\{N-\ell,k\}}(1-y^{M_\ell})\bigg|_p\nonumber\\
&=&\bigg|(1-y)\sum_{\ell=0}^NC_N^\ell x^\ell y^{\min\{N-\ell,k\}}\big(\sum_{j=0}^{M_\ell}y^j\big)\bigg|_p\nonumber\\
&\leq& |1-y|_p\max\limits_{0\leq\ell\leq N}\bigg\{\bigg|C_N^\ell x^\ell y^{\min\{N-\ell,k\}}\big(\sum_{j=0}^{M_\ell}y^j\big)\bigg|_p\bigg\}\nonumber\\
&\leq&|1-y|_p,
\end{eqnarray}
here $M_\ell=\max\{N-\ell,k\}-\min\{N-\ell,k\}$.

From \eqref{e-dif-1} with \eqref{e-dif-12} we immediately find
\begin{eqnarray}\label{e-dif-2}
|\m_0(\s)-\m_i(\s)|_p\leq\frac{|h_0|_p^2|1-x_i|_p}{p^{H(\s)}} \ \ \
(i=1,2).
\end{eqnarray}

Using the same argument we get
\begin{eqnarray}\label{e-dif-3}
|\m_1(\s)-\m_2(\s)|_p\leq\frac{|h_0|_p^2|x_1-x_2|_p}{p^{H(\s)}}.
\end{eqnarray}

Consequently, we can formulate the following

\begin{thm}\label{bound1}
Assume that $|q|_p=1$ and  the measures $\m_1$ $\m_2$ for $p$-adic
the ferromagnetic $q+1$-state Potts model \eqref{Potts} exist. Then
the measures $\m_k$ ($k=0,1,2$) are bounded. Moreover, the
inequalities \eqref{e-dif-2},\eqref{e-dif-3} hold. In this case,
there is a quasi phase transition.
\end{thm}

\subsection{Antiferromagnetic case} In this case we assume that $N<0$ and $-N$ is even. Then
according to Theorem \ref{aq2} there exist the measures
$\m_0$,$\m_1$ and $\m_2$.

Now taking into account \eqref{ax-12-q2},\eqref{ax1-q} from
\eqref{e-mu1} one finds
\begin{eqnarray}\label{ae-mu1}
|\m_1(\s)|_p&=&\frac{|h_0|_p^2}{|x_1|_p^{2|V_{n-1}|}}\cdot\frac{|x_1|_p^{\#\s}}{p^{H(\s)}}\nonumber\\
&=&\frac{|h_0|_p^2\cdot p^{-2\bar N(2|V_{n-1}|-\#\s)}}{p^{H(\s)}}\nonumber \\
&=&|h_0|_p^2\cdot p^{-2\bar N\big(2|V_{n-1}|-\#\s-\frac{1}{2}\sum\limits_{<x,y>\in L_n}\delta_{\s(x),\s(y)}\big)}
\end{eqnarray}

Now let us estimate the expression standing inside the brackets. It is clear that
\begin{equation}\label{aconf1}
0\leq\#\s\leq|W_n|, \quad 0\leq \sum\limits_{<x,y>\in L_n}\delta_{\s(x),\s(y)}\leq |V_n|-1.
\end{equation}
Therefore, from \eqref{aconf1} with $|W_n|-|V_{n-1}|=2$,
$|V_n|=|V_{n-1}|+|W_n|$ we get
\begin{eqnarray}\label{ae-mu12}
2|V_{n-1}|-\#\s-\frac{1}{2}\sum\limits_{<x,y>\in L_n}\delta_{\s(x),\s(y)}&\geq&
2|V_{n-1}|-|W_n|-\frac{1}{2}(|V_n|-1)\nonumber\\
&=&|V_{n-1}|-2+\frac{1}{2}(1-|V_n|)\nonumber\\
&=&\frac{1}{2}\big(2|V_{n-1}|-|V_n|-3\big)\nonumber\\
&=&\frac{1}{2}\big(|V_{n-1}|-|W_n|-3\big)\nonumber\\
&=&-\frac{5}{2}
\end{eqnarray}
Consequently, the last inequality \eqref{ae-mu12} with \eqref{ae-mu1} implies
\begin{eqnarray}\label{ae-mu13}
|\m_1(\s)|_p\leq |h_0|_p^2\cdot p^{5\bar N}=\frac{|h_0|_p^2}{p^{5N}},
\end{eqnarray}
this means that $\m_1$ is bounded.

Now consider the measure $\m_2$. Noting $|x_2|_p=|q|_p^2 p^{-2\bar
N}$ and $|x_2+\t+q-1|_p=p^{\bar N}$ (see
\eqref{ax-12-q2},\eqref{ax2-q}), the  equality \eqref{e-mu2} yields
\begin{eqnarray}\label{ae-mu2}
|\m_2(\s)|_p&=&\frac{|q|_p^2|h_0|_p^2}{p^{2\bar N|V_{n-1}|}}\cdot\frac{p^{-2\bar N{\#\s}}}{p^{H(\s)}}\nonumber\\
&=& |q|_p^2|h_0|_p^2\cdot p^{-2\bar N\big(|V_{n-1}|-\frac{1}{2}\sum\limits_{<x,y>\in L_n}\delta_{\s(x),\s(y)}+\#\s\big)}.
\end{eqnarray}

Now using the same argument as in \eqref{ae-mu12} one gets
\begin{eqnarray}\label{ae-mu21}
|V_{n-1}|-\frac{1}{2}\sum\limits_{<x,y>\in L_n}\delta_{\s(x),\s(y)}+\#\s\geq
|V_{n-1}|-\frac{1}{2}(|V_n|-1)=-\frac{1}{2}.
\end{eqnarray}
Hence, \eqref{ae-mu21} with \eqref{ae-mu2} implies
\begin{eqnarray}\label{ae-mu22}
|\m_2(\s)|_p\leq \frac{|h_0|_p^2}{p^{N}},
\end{eqnarray}
which means that $\m_2$ is bounded as well.

Let us consider the measure $\m_0$. From \eqref{e-mu0} we obtain
\begin{eqnarray}\label{ae-mu0}
|\m_0(\s)|_p&=&\frac{|h_0|_p^2}{|\t+q|_p^{2|V_{n-1}|}}\cdot\frac{1}{p^{H(\s)}}\nonumber\\
&=&|h_0|_p^2p^{-\bar N(2|V_{n-1}|-\sum\limits_{<x,y>\in L_n}\delta_{\s(x),\s(y)})}\nonumber\\
&\leq & |h_0|_p^2p^{\bar N}\nonumber\\
&=&\frac{|h_0|_p^2}{p^{N}},
\end{eqnarray}
here we have used (see \eqref{ae-mu21})
$$
2|V_{n-1}|-\sum\limits_{<x,y>\in L_n}\delta_{\s(x),\s(y)}\geq -1.
$$
Hence, $\m_0$ is bounded too.

Let us consider relations between $\m_0$ and $\m_1,\m_2$. From \eqref{ae-mu0},\eqref{ae-mu1} and \eqref{ae-mu2} we find
\begin{eqnarray}\label{ae-mu-01}
&&\frac{|\m_1(\s)|_p}{|\m_0(\s)|_p}=p^{-2\bar N(|V_{n-1}|-\#\s)}\leq p^{4\bar N},\\[2mm]
\label{ae-mu-02}
&&\frac{|\m_2(\s)|_p}{|\m_0(\s)|_p}=|q|_p^2\cdot p^{-\bar N(\#\s)}\leq |q|^2_p,
\end{eqnarray}
here in \eqref{ae-mu-01} we have used $|V_{n-1}|-\#\s\geq |V_{n-1}|-|W_n|=-2$.
Hence, the derived relations imply that
\begin{equation}\label{ae-mu-00}
|\m_1(\s)|_p\leq p^{4\bar N}|\m_0(\s)|_p, \quad |\m_2(\s)|_p\leq |q|^2_p|\m_0(\s)|_p.
\end{equation}
Let us consider relation between $\m_1$ and $\m_2$. From \eqref{ae-mu1} and \eqref{ae-mu2} we find
\begin{eqnarray}\label{ae-mu-1}
\frac{|\m_1(\s)|_p}{|\m_2(\s)|_p}=\frac{p^{-2\bar N(|V_{n-1}|-2\#\s)}}{|q|_p^2}.
\end{eqnarray}
Take any configuration $\s_n$ in $\Om_{V_n}$ with $\#\s_n=|W_n|$, (for example
$\s_n(x)=1$ for every $x\in V_n$). Then \eqref{ae-mu-1} yields
\begin{eqnarray}\label{ae-mu-2}
\frac{|\m_1(\s_n)|_p}{|\m_2(\s_n)|_p}\geq p^{2\bar N(|W_n|-2)}\to \infty \ \ \textrm{as} \ n\to\infty.
\end{eqnarray}

Now take any configuration $\tilde\s_n$ in $\Om_{V_n}$ with $\#\tilde\s_n=0$, (for example
$\tilde\s_n(x)=0$ for every $x\in V_n$). Then \eqref{ae-mu-1} yields
\begin{eqnarray}\label{ae-mu-3}
\frac{|\m_1(\tilde\s_n)|_p}{|\m_2(\tilde\s_n)|_p}=\frac{p^{-2\bar N|V_{n-1}|}}{|q|_p^2}\to 0  \ \textrm{as} \ n\to\infty.
\end{eqnarray}

The relations \eqref{ae-mu-1},\eqref{ae-mu-2} show that the structure of the measures $\m_1$ and $\m_2$ are different even they are bounded.

Consequently, we can formulate the following

\begin{thm}\label{bound2}
Let $N<0$ and $-N$ is even. Then  the translation-invariant $p$-adic
quasi Gibbs measures $\m_0$, $\m_1$ and $\m_2$ of antiferromagnetic
Potts model \eqref{Potts} are bounded.  Moreover, the inequality
\eqref{ae-mu-00} holds. In this case, there is a quasi phase
transition.
\end{thm}

\section{Conclusions}

It is known that to investigate phase transitions, a dynamical
system approach, in real case, has greatly enhanced our
understanding of complex properties of models. The interplay of
statistical mechanics with chaos theory has even led to novel
conceptual frameworks in different physical settings \cite{E}.
Therefore, in the present paper, we have investigated a phase
transition phenomena from such a dynamical system point of view. For
$p$-adic quasi Gibbs measures of $q+1$-state Potts model on a Cayley
tree of order two,  we derived a recursive relations with respect to
the boundary conditions, then we defined one dimensional fractional
$p$-adic dynamical system. In ferromagnetic case, we have
established that if $q$ is divisible by $p$, then such a dynamical
system has two repelling and one attractive fixed points. We found
basin of attraction of the attractive fixed point, and this allowed
us to describe all solutions of the nonlinear recursive equations.
Moreover, in that case we prove the existence of the strong phase
transition. If $q$ is not divisible by $p$, then the fixed points
are neutral, and the existence of the quasi phase transition has
been established. In antiferromagnetic case, there are two
attractive and one repelling fixed points. We found  basins of
attraction of both attractive fixed points, and described solutions
of the nonlinear recursive equation. In this case, we proved the
existence of a quasi phase transition as well. These investigations
show that there are some similarities with the real case, for
example, the existence of two repelling fixed points implies the
occurrence of the strong phase transition. Moreover, using such a
method one can study other $p$-adic models over trees.

 Note that the
obtained results are totaly different from the results of
\cite{MR1,MR2}, since when $q$ is divisible by $p$ means that $q+1$
is not divided by $p$, which according to \cite{MR1} means that
uniqueness and boundedness of $p$-adic Gibbs measure.

\section*{Acknowledgement} The present study have been done within
the grant FRGS0409-109 of Malaysian Ministry of Higher Education.

\end{document}